\documentclass[letterpaper, 11pt]{article}

\usepackage{graphicx}
\usepackage{lipsum} 
\usepackage{amsmath}
\usepackage{amsthm}
\usepackage{mathtools}
\usepackage{amsfonts}
\usepackage{amssymb}
\usepackage{empheq} 
\usepackage{dashbox}
\usepackage{booktabs}

\usepackage[letterpaper, margin=1in]{geometry}
\usepackage[section]{placeins}

\newtheorem{lemma}{Lemma}
\newtheorem{theorem}{Theorem}
\newtheorem{definition}{Definition}

\newtheorem{claim}{Claim}
\newtheorem{corollary}{Corollary}

\newcommand{\changes}[1]{{\color{black} #1}}
\newcommand{\camera}[1]{{\color{black} #1}}

\usepackage{amsmath}
\usepackage{bbm}

\usepackage[most]{tcolorbox}

\usepackage{optidef}
\usepackage{xcolor}
\newcommand\mycommfont[1]{\ttfamily\textcolor{gray}{#1}}

\usepackage{colortbl}
\usepackage{enumitem}

\newcommand{\poly}{\textup{poly}}

\newcommand{\cI}{{\mathcal{I}}}
\newcommand{\cX}{{\mathcal{X}}}
\newcommand{\cY}{{\mathcal{Y}}}
\newcommand{\cS}{{\mathcal{S}}}
\newcommand{\cF}{{\mathcal{F}}}
\newcommand{\cR}{{\mathcal{R}}}

\newcommand{\cT}{{\mathcal{T}}}
\newcommand{\cC}{{\mathcal{C}}}
\newcommand{\diver}{\textup{div}}
\newcommand{\lsf}{\ell^*}

\newcommand{\divScore}{\ell^*}

\newcommand{\appr}{\epsilon}
\newcommand{\foo}{\epsilon'}

\newcommand{\doubldim}{\lambda}
\newcommand{\dimension}{D}

\usepackage{xspace}
\newcommand{\gonzalez}{\textsc{GMM}\xspace}
\newcommand{\GMM}{\textsc{GMM}\xspace}
\newcommand{\tauGMM}{$\tau$-\textsc{GMM}\xspace}

\newcommand{\tauGMMstream}{$\tau$-\textsc{GMM-Stream}\xspace}
\newcommand{\tilk}{\widetilde{k}}
\newcommand{\fMM}{\textsc{Fair Max-Min}\xspace}

\newcommand{\fMMFlow}{\textsc{Fair-Flow}\xspace}

\newcommand{\fGMM}{\textsc{Fair-GMM}\xspace}
\newcommand{\fDP}{\textsc{Fair-DP}\xspace}
\newcommand{\fL}{\textsc{Fair-Line}\xspace}
\newcommand{\fB}{\textsc{Fair-Euclidean}\xspace}
\newcommand{\randround}{\textsc{Randomized-Rounding}}

\newcommand{\genswap}{\textsc{Fair-Greedy-Flow}\xspace}
\newcommand{\streamrbcol}{\textsc{Fair-Stream-2Groups}\xspace}

\newcommand{\ball}[1]{\mathbf{B}(#1, \gamma/2)}
\newcommand{\rball}[2]{\mathbf{B}(#1, #2)}
\newcommand{\dmin}{d_\text{min}}
\newcommand{\dmax}{d_\text{max}}
\newcommand{\tb}{H}

\usepackage{algorithm}
\usepackage{algpseudocode}

\algnewcommand{\LineComment}[1]{\State{\mycommfont{$\triangleright${#1}}}}
\algnewcommand{\LineCommentx}[1]{\Statex{\mycommfont{$\triangleright${#1}}}}

\usepackage{setspace, stackengine}

\DeclareMathOperator*{\E}{{\bf {E}}}

\title{	
\LARGE Improved Approximation and Scalability \\ 
for Fair Max-Min Diversification \footnote{University of Massachusetts Amherst, \{raddanki, mcgregor, ameli, zmoumoulidou\}@cs.umass.edu}
}

\author{
  \large Raghavendra Addanki \quad  Andrew McGregor \quad Alexandra Meliou \quad Zafeiria Moumoulidou \\
  } 
\date{} 

\begin{document}
\maketitle

\begin{abstract}
Given an $n$-point metric space $(\cX,d)$ where each point belongs to one of
$m=O(1)$ different categories or groups and a set of integers $k_1, \ldots, k_m$, the fair Max-Min diversification problem is to select $k_i$ points belonging to category $i\in [m]$, such that the minimum pairwise distance between selected points is maximized. The problem was introduced by Moumoulidou et al.~[ICDT 2021] and is motivated by the need to down-sample large data sets in various applications so that the derived sample achieves a balance over
\emph{diversity}, i.e., the minimum distance between a pair of selected points,
and \emph{fairness}, i.e., ensuring enough points of each category are
included. We prove the following results:

\begin{enumerate}
	\item We first consider general metric spaces. We present a randomized polynomial time algorithm that returns a factor $2$-approximation to the diversity but only satisfies the fairness constraints in expectation. Building upon this result, we present a $6$-approximation that is guaranteed to satisfy the fairness constraints up to a factor $1-\epsilon$ for any constant $\epsilon$. We also present a linear time algorithm returning an $m+1$ approximation with exact fairness. The best previous result was a $3m-1$ approximation.
	\item We then focus on Euclidean metrics. We first show that the problem can be solved \emph{exactly} in one dimension. {For constant dimensions, categories
and any constant $\epsilon>0$, we present a $1+\epsilon$ approximation
algorithm that runs in $O(nk) + 2^{O(k)}$ time where $k=k_1+\ldots+k_m$.} We
can improve the running time to $O(nk)+\poly(k)$ at the expense of only picking
$(1-\epsilon) k_i$ points from category $i\in [m]$.
\end{enumerate}

Finally, we present algorithms suitable to processing massive data sets
including single-pass data stream algorithms and composable coresets for the
distributed processing.
\end{abstract}


\section{Introduction}\label{sec:Intro}

Given a universe of $n$ elements $\cX$ and a metric distance function $d: \cX
\times \cX \rightarrow \mathbb{R}_{0}^{+}$, the Max-Min diversification problem
seeks to select a $k$-sized subset $\cS$ of $\cX$ such that the minimum
distance between the points in $\cS$ is maximized~\cite{chandra2001, ravi1994}.
Intuitively, the goal is to maximize the \emph{dissimilarity} across all the
selected points while $k$ is typically much smaller than $n$. A considerable
amount of work in the database community has addressed the diversity
maximization problem in the context of query result
diversification~\cite{deng2014complexity, drosoudisc, viera11}, efficient
indexing schemes for result diversification~\cite{Agarwal2020,
Drosou2014, Wang2018}, nearest neighbor search~\cite{abbar2013diverse}, ranking
schemes~\cite{AngelK2011, QinYC2012}, and  recommendation systems~\cite{abbassi2013, Borodin2012}.

Recently, Moumoulidou et al.~\cite{moumoulidou21} introduced the \emph{fair}
variant of the Max-Min diversification problem. Specifically, the assumption is
that the universe of elements $\cX$ is partitioned into $m= O(1)$ disjoint
categories or groups. Then, the aim is to construct a diverse set of points
where each group is sufficiently represented. To this end, the input of the
problem includes non-negative integers $k_1,\dots ,k_m$ and the goal now is to select a subset $\cS$ using $k_i$ representatives from each group such that the minimum distance across all points is maximized. As a concrete example, consider a query over a maps service for finding restaurants around Manhattan at NYC. Then the goal is to present the user with a diversified set of restaurant locations while representing different cuisines in the sample.

In this work, we improve currently known approximation results for fair
Max-Min diversification. This includes improving the approximation factor in the most general case of the problem; significantly decreasing the approximation factor if we slightly relax the fairness constraints; and reducing the approximation factors to arbitrarily close to 1 when the underlying metric is Euclidean. Before presenting our results, we review related work.

\subsection{Related Work} The problem of unconstrained diversity maximization, \camera{i.e., when the number of groups $m=1$,} is well-studied in the context of facility location, information retrieval, web search and recommendation systems~\cite{AngelK2011, mmr, chandra2001, drosoudisc, ERKUT199048, gollapudi2009, hassin1997, Kuby1987, QinYC2012, ravi1994, viera11}. We refer the interested readers to the following surveys related to the diversification literature~\cite{article, drosou2010}.

Among popular diversification models are the \emph{distance-based} models. In
these models, the diversity of a set of points is modeled via some function
defined over pairwise distances. Max-Sum (also known as \emph{remote-clique})
and Max-Min (also known as \emph{remote-edge} or \emph{p-dispersion}) are two
of the most well-established distance-based diversification models~\cite{Indyk14}. In Max-Sum,
diversity is defined as the sum of the pairwise distances of points selected in
a set, while in Max-Min the diversity of a set is equal to the minimum pairwise
distance. For both problems, there are known $2$-approximation algorithms,
which yield the best approximation guarantee that can be achieved for both
problems~\cite{Bhaskara16, Borodin2012, ravi1994}. There are also recent works
on \emph{distance-based} diversity maximization models in the streaming,
distributed, and sliding-window models~\cite{Aghamolaei2015, Borassi19, Ceccarello17, Indyk14}. 

Contrary to unconstrained diversity maximization, the problem of fair
diversity maximization is less studied. To the best of our knowledge, there is
a known 2-approximation local search algorithm for fair Max-Sum
diversification~\cite{abbassi2013, Borodin2017, Borodin2012} where fairness is
modeled via partition matroids~\cite{schrijver2003}. Recent work also extends
the local search approach to distances of negative type~\cite{Cevallos17}. Another recently studied objective called Sum-Min~\cite{Bhaskara16} is defined as the sum of distances of all points to their closest point in the set. Bhaskara et al.~\cite{Bhaskara16} present an $8$-approximation algorithm for Sum-Min under partition matroid constraints.

The most relevant result to our work is due to Moumoulidou {et al.}~\cite{moumoulidou21} that introduced the fair variant for the Max-Min
diversification problem that we also study. \camera{The proposed fairness objectives have been widely studied by prior work~\cite{bera2019fair, pmlr-v80-celis18a, Celis2017RankingWF, chierichetti2017fair, chiplunkar2020solve, halabi2021, Jones2020FairKV, pmlr-v97-kleindessner19a, Stoyanovich2018OnlineSS, wang2021, Yang:2019:BRD:3367722.3367886, YangS2017, Zehlike:2017:FFT:3132847.3132938}, and are based on the definition of group fairness and statistical parity~\cite{Dwork:2012:FTA:2090236.2090255}. It is worth noting that there are other definitions for fairness, like individual or causal fairness~\cite{GalhotraBM2017}, but these are not the focus of our work.} Moumoulidou {et al.}~\cite{moumoulidou21} designed a polynomial time algorithm that achieved a $3m-1$- approximation for fair Max-Min diversification. There is also a recent line of work for designing (composable) coresets for various distance-based diversification objectives in the \camera{fairness} setting~\cite{Ceccarello2018, Ceccarello2021}. Coresets are small subsets of the original data that contain a good approximate solution and are typically used for speed up purposes or designing streaming and distributed algorithms. Prior efforts leave as an open question the construction of coresets for the fair variant of the Max-Min diversification objective. 

\subsection{Our Results}
We present results for both the cases of general metrics and Euclidean metrics.
\begin{enumerate}
	\item{\bf General Metrics.} In Section~\ref{subsec:randomLP}, we 	present a randomized polynomial time algorithm that returns a factor 2-approximation to the diversity but only satisfies the fairness constraints in expectation, i.e., for each $i\in [m]$, the output is expected to include at least $k_i$ points from $\cX_i$. In Section~\ref{subsec:randomLP2},  we present a 6-approximation that is guaranteed to include $(1-\epsilon)k_i$ points in each group $i\in [m]$ assuming each $k_i=\Omega(\epsilon^{-2} \log m)$. Both these results are based on randomized rounding of a linear program. Finally, in Section~\ref{subsec:lineartime} we present a linear time algorithm returning an $m+1$ approximation with perfect fairness. \changes{This is an improvement over the previously known $3m-1$ approximation~\cite{moumoulidou21}. We present an example in Appendix~\ref{app:tightexample} that shows that the analysis presented in Moumoulidou et al.~\cite{moumoulidou21} cannot be improved to obtain a better approximation factor.} In Section~\ref{subsec:hardness}, we present a hardness of approximation result arguing that we cannot get an approximation factor better than $2$, even allowing for multiplicative approximations in fairness constraints.
	\item {\bf Euclidean Metrics.} If the points can be embedded in low dimensional space ${\mathbb R}^D$ (e.g., if the points correspond to geographical locations) and the distances correspond to Euclidean distances then we can significantly improve the approximation factors of our algorithms. In Section~\ref{subsec:line}, we show that the problem can be solved \emph{exactly} for $D=1$. For constant dimensions, groups, we then present a $1+\epsilon$ approximation algorithm that runs in $O(nk)+2^{O(k)}$ time \camera{where $k = k_1 + k_2 + \ldots+k_m$}. In Section~\ref{subsec:bicriteria}, we show how to improve the running time to $O(nk) + \poly(k)$ at the expense of only picking $(1-\epsilon) k_i$ points from group $i\in [m]$. All these results are based on a new coreset construction.
\end{enumerate}

In Sections~\ref{subsec:streams} and \ref{sec:composable_coresets}, we present algorithms suitable to processing massive data sets
including single-pass data stream algorithms and composable coresets for distributed processing.

\section{Background and Preliminaries} \label{sec:Background}

\subsection{Fair Max-Min Diversification}
We formally define the problem of fair Max-Min diversification recently
introduced in~\cite{moumoulidou21}.

\begin{definition}[\fMM] \label{def:fMax}
Let $(\cX, d)$ be a metric space where $\cX = \bigcup_{i=1}^{m}\cX_i$ is a
universe of $n$ elements partitioned into $m$ non-overlapping groups and $d:
\cX \times \cX \rightarrow \mathbb{R}_{0}^{+}$ is a metric distance function.
Then $\forall u, v \in \cX$, $d$ satisfies the following properties: (1)~$d(u,
v) = 0$ iff $u=v$ (identity), (2)~$d(u, v) = d(v, u)$ (symmetry), and (3)~$d(u,
v) \leq d(u, w) + d(w, v)$ (triangle inequality). Further, let $k_1, k_2,
\cdots, k_m$ be non-negative integers with $k_i \leq |\cX_i|, \ \forall i \in
[m]$. The problem of fair Max-Min diversification is now defined as follows:
\begin{maxi*}|l|
{\cS \subseteq \cX}{\min_{\substack{u,v\in \mathcal{S}\\ u\neq v}}d(u,v)} {}{}
\addConstraint{|\mathcal{S} \cap
\mathcal{X}_i| = k_i, \ \forall i \in [m] \quad \textup{(fairness constraints)}} 
\end{maxi*}\end{definition}

The aim is to select a subset $\cS \subseteq \cX$ of points that maximizes the
minimum pairwise distance across the points in $\cS$ while being constrained to
include $k_i$ points from group $i$. Throughout the paper we refer to the
diversity of a set $\cS$ as $\diver(\cS) = \min_{u,v\in \mathcal{S}, u\neq
v}d(u,v)$.

Let $\cS^{*} = \bigcup_{i=1}^{m} \cS^{*}_i$ be the set of points that obtains
the optimal diversity score denoted by $\diver{(\cS^{*})}=\divScore$. We say a
subset of points $\cS$ is an $\alpha$ approximation if $\diver(\cS)\geq
\lsf/\alpha$ and achieves $\beta$ fairness if $|\cS \cap \cX_i|\geq \beta k_i$
for all $i\in [m]$. When $\beta=1$, we say subset achieves \emph{perfect
fairness}.

\fMM is an NP-hard problem for which the best known polynomial time algorithms
are: a $4$-approximation algorithm that only works for $m=2$ groups and a
$3m-1$-approximation algorithm that yields the best guarantees for any $m \geq
3$~\cite{moumoulidou21}. The best approximation factor one can hope for in
general metric spaces is a $2$-approximation guarantee. This claim easily
follows since when $m=1$, the problem is just the Max-Min diversification
problem where it is known that no polynomial time algorithm \camera{with an approximation factor better than 2} exists if $P\neq NP$~\cite{ravi1994}. \changes{We use $\poly(\cdot)$ to describe polynomial time algorithms using the context dependent parameters.}

\subsection{Low Doubling Dimension Spaces}

Our results for low dimensional Euclidean metrics use the fact that such metrics
have \emph{low doubling dimension}. Our work in this direction is inspired by
work on diversity maximization by Ceccarello et
al.~\cite{Ceccarello2018,Ceccarello2021,Ceccarello17}. We define a ball of
radius $r$ centered at $p \in \cX$ as the set of all points in $\cX$ within
distance strictly less than $r$ from $p$. We use the notation: $\rball{p}{r} =
\{ q \in \cX \mid d(p, q) < r \}.$

\begin{definition}[\textsc{Doubling Dimension}]\label{def:doubling} 
Let $(\cX, d)$ be a metric space. The \emph{doubling dimension} of $\cX$ is the
smallest integer $\doubldim$ such that any ball $\rball{p}{r}$ of radius $r$
around a point $p \in \cX$ can be covered using at most $(r/r')^\doubldim$
balls of radius $r'$. The Euclidean metric on $\mathbb{R}^{D}$ has doubling
dimension $O(D)$~\cite{assouad1983plongements, Ceccarello17, gupta2003bounded}.
\end{definition}

\subsection{Coresets}
Coresets are powerful theoretical tools for designing efficient optimization
algorithms in the presence of massive datasets in sequential, streaming or distributed environments~\cite{Agarwal2007GeometricAV,
Indyk14}. At a high level, coresets are carefully chosen subsets of the original universe of elements that contain an approximate solution to the optimal solution for the optimization problem at hand. A coreset for fair Max-Min diversification is defined as follows:

\begin{definition}[\textsc{Coreset for} \fMM]
A set $\cT \subseteq \cX$ is an $\alpha$-coreset if there exists a subset $\cT^{'} \subseteq \cT$ with $|\cT^{'} \cap \cX_i|=k_i \ \forall i \in [m]$ and $\diver{(\cT^{'})} \geq {\divScore}/{\alpha}$.
\end{definition}

Note that optimally solving \fMM on $\cT$,
a set typically much smaller in size than $\cX$, yields an $\alpha$-approximation factor. Further, the notion of coresets is useful for designing algorithms in the distributed setting using the \emph{composability} property. \emph{Composable coresets} closely relate to the notion of \emph{mergeable summaries}~\cite{10.1145/2213556.2213562, Indyk14} while the assumption is that the universe of elements $\cX$ is partitioned into $L$ subsets (e.g., processing sites). Then the goal is to process each subset independently and extract a \emph{local} coreset such that in the union of these local coresets, there is an approximate solution for the optimization problem at hand. Specifically, for \fMM a composable coreset is defined as follows:

\begin{definition}[\textsc{Composable coreset for} \fMM]
A function $c(\cX)$ that maps a set of elements to a subset of these
elements computes an $\alpha-$composable coreset for some $\alpha \geq 1$, if
for any partitioning\footnote{The notion of composable coresets can also be
extended when $\cX$ is not divided into disjoint subsets but this is not the
focus of our work.} of $\cX=\bigcup_j \cY_j$ and $\cT= \bigcup_{j} c(\cY_j)$,
there exists a set $\cT^{'} \subseteq \cT$ with $|\cT^{'} \cap \cX_i|=k_i \
\forall i \in [m]$ such that $\diver{(\cT^{'})} \geq {\lsf}/{\alpha}$.
\end{definition}

\section{General Metrics}\label{sec:genmetrics}

In this section, we present algorithms for \fMM with an arbitrary metric. Our first two algorithms are based on rounding a suitable linear program. In Section~\ref{subsec:lineartime} we present a linear time algorithm returning an $m+1$ approximation with perfect fairness. Finally, in Section~\ref{subsec:hardness}, we give hardness of approximation results for \fMM.


\subsection{2-Approx with Expected Fairness}\label{subsec:randomLP}
In this section and others, we assume a guess $\gamma$ on the optimal diversity value for \fMM. Note there are at most $\binom{n}{2}$ possible values for the optimal diversity corresponding to the set of distances between pairs of points. Hence, trying all these guesses only increases the running time by a factor $O(n^2)$. Assuming the ratio between the largest and smallest distance is $\poly(n)$, this can be reduced to $O(\epsilon^{-1} \log n)$ at the expense of introducing an additional factor of  $1+\epsilon$ in the approximation. This follows by the standard technique of only considering guesses that are powers of $(1+\epsilon)$~\cite{Guha2009}. 

\medskip
\noindent \textbf{\fMM \textsc{LP}}.
Let $\cX=\{p_1,\ldots, p_n\}$. For every point $p_j \in \cX$, we have a
variable $x_j$. We represent the fairness constraint for every group $i \in
[m]$ using constraint~(\ref{eq:constraint1}). Additionally, for every point $p
\in \cX$, we add the constraint~(\ref{eq:constraint2}) that includes at most
one point in a ball of radius $\gamma/2$ centered at $p$. This ensures that the
selected points are separated by a distance of at least $\gamma/2$. Using
constraint (\ref{eq:constraint3}), we allow $x_p$ to take any value between $0$
and $1$. If $\gamma \le \lsf$, observe that the optimal solution for \fMM is a feasible solution for this LP.
\begin{align*}
    \sum_{p_j \in \mathcal \cX_i} x_j &\geq k_i \quad \forall i \in [m]. \tag{1} \label{eq:constraint1}\\
    \sum_{p_\ell \in \ball{p}} x_\ell &\leq 1 \quad \forall p \in \cX. \tag{2} \label{eq:constraint2}\\
    x_j &\geq  0 \quad \forall j\in [n]. \tag{3} \label{eq:constraint3}
\end{align*}
Let $x^*_j$ denote the optimal solution of the linear program stated above. Let
$n'=|\{j: x^*_j>0\}|$ and without loss of generality suppose $x^*_j>0$ for all
$j\in [n']$. \changes{We obtain an integral solution using a randomized rounding algorithm, in which we generate a random ordering based on sampling without replacement, such that a point $p_j$ is selected as the next point in the ordering with probability proportional to $x^*_j$. This allow us to show (see Lemma~\ref{lem:randround}) that the rounding scheme returns a set $\cS$ with at least $k_i$ points \emph{in expectation} from each group $i \in [m]$ (satisfying constraint~(\ref{eq:constraint1}) in expectation). Further, our rounding scheme selects at most one point from each ball of radius $\gamma/2$ (satisfying constraint~(\ref{eq:constraint2})). \camera{Since for a $\gamma\leq \lsf$ there is a set $\cS$ that satisfies the properties discussed above, selecting the set $\cS$ for the largest guess $\gamma$} results in a $2$-approximation for the diversity score.}

\medskip
\noindent \textbf{Randomized Rounding}.
We generate a random ordering $\sigma$ of $[n']$ where $\sigma(t)$ is randomly
chosen from $R_{t}=[n']\setminus \{\sigma(1), \ldots, \sigma(t-1)\}$ such that
for $j\in R_{t}$, $ \Pr[\sigma(t)=j]= \frac{x^*_j}{\sum_{\ell\in R_{t}
}x^*_\ell} \ .$ After generating the ordering $\sigma$, we construct the output set $\cS$ by including the point $p_j$ in $\cS$ iff $\sigma(j)\leq
\sigma(\ell)$ for all $p_\ell \in \rball{p_j}{\gamma/2}$. Note that all points in the output are at least distance $\gamma/2$ apart.

\begin{lemma}\label{lem:randround}
There is an algorithm that returns a set $\cS$, such that
for all groups $i \in [m]$, it holds that $\mathbb{E}[|\cS \cap \cX_i|]\geq
k_i$. Further all the points selected in $\cS$ are at least $\gamma/2$ far
apart.
\end{lemma}

\begin{proof}
Consider the randomized rounding algorithm described in this section. Now, let $p_j$ be a point with $x^*_j>0$. Define
$A_t$ to be the event $d(p_{\sigma(t)},p_j)<\gamma/2$ and $d(p_{\sigma(t')},p_j)\geq \gamma/2$ for all $t'<t$. In other words, $A_t$ is the event that the first point included in $\cS$ from the ball $\rball{p_j}{\gamma/2}$ is the point from the $t$-th step (in the ordering $\sigma$). Then,

\begin{align*}
 \Pr[p_j\in \cS]  = \sum_{t=1}^{n'} \Pr[\sigma(t)=j |A_t] \Pr[A_t] 
    &= \sum_{t=1}^{n'}\frac{x^*_j}{\sum_{p_\ell \in \rball{p_j}{\gamma/2}} x^*_\ell } \Pr[A_t]\\ &= \frac{x^*_j}{\sum_{p_\ell \in \rball{p_j}{\gamma/2}} x^*_\ell } \sum_{t=1}^{n'}  \Pr[A_t] \\
    &= \frac{x^*_j}{\sum_{p_\ell \in \rball{p_j}{\gamma/2}} x^*_\ell }    \geq  x^*_j
    \end{align*}
	 
where the last equality follows because $\sum_{t=1}^{n'}  \Pr[A_t]=1$ and the last inequality holds because of constraint (\ref{eq:constraint2}) in the
\fMM LP. 
Then for $i \in [m]$, we have $\mathbb{E}[|\cS \cap \cX_i|] \geq \sum_{p \in \cX_i} x^*_p \geq k_i$ where the last inequality follows from constraint (\ref{eq:constraint1}).
\end{proof}

\subsection{6-Approx with $(1-\epsilon)$ Fairness}\label{subsec:randomLP2}

We now present a more involved rounding scheme of the LP given in the previous section that ensures that the selected points contain at least $(1-\epsilon)k_i$ points in $\cX_i$ for each $i\in [m]$. However, this guarantee comes at the expense of increasing the approximation factor for the diversity score from 2 to 6. 

The main idea behind the new rounding scheme stems from the observation that for any $p_i, p_j\in \cX$, if $\ball{p_i}$ and $\ball{p_j}$ are disjoint, then, in the previous rounding scheme, the event that $p_i$ is included in the returned solution is independent of the event that $p_j$ is  included. This follows because  the relative  ordering of the elements in $\{\ell: p_\ell\in \ball{p_i}\}$ in $\sigma$ is independent of the ordering of the elements in $\{\ell: p_\ell\in \ball{p_j}\}$ in $\sigma$. This independence will ultimately allow us to use Chernoff bound to argue concentration of the number of elements chosen from each group $\cX_j \ \forall j \in [m]$.

\subsubsection{Randomized Rounding with improved fairness guarantees}
We solve the LP in Section~\ref{subsec:randomLP} to get a feasible solution $\{x^*_j\}_{j\in [n]}$. Next, we transform $\{x^*_j\}_{j\in [n]}$ into a feasible solution $\{y_j^*\}_{j\in [n]}$ for the following set of constraints, some of which are no longer linear:
\begin{align*}
    \sum_{p_j \in \mathcal \cX_i} y_j  \geq k_i & \quad \forall i \in [m]. \tag{1'} \label{eq:constraint1new}\\
    \sum_{p_\ell \in \rball{p}{\gamma/6}} y_\ell \leq 1 &\quad \forall p \in \cX. \tag{2'} \label{eq:constraint2new}\\
    y_j \geq 0 & \quad  \forall j\in [n]. \tag{3'} \label{eq:constraint3new} \\ 
(0<y_i \mbox{ and } 0<y_j)  \Rightarrow  d(p_i, p_j)  \geq \frac{\gamma}{3} & \quad \forall p_i, p_j \in \cX_\ell  , \forall \ell\in [m]  \tag{4'} \label{eq:constraint4new} \end{align*}

\changes{The constraint~(\ref{eq:constraint2new}) ensures that at most one point in a ball of radius $\gamma/6$ is selected (instead of $\gamma/2$ used in Section~\ref{subsec:randomLP}) and results in an approximation factor of $6$. The constraint~(\ref{eq:constraint4new}) ensures that points from the same group with non-zero values are separated by at least $\gamma/3$, which is used to argue $(1-\epsilon)$ fairness (see Theorem~\ref{thm:6approx})}. The transformation of $x^*$ to $y^*$ can be done by redistributing the values as follows:
\begin{itemize}
    \item[(a)]  For each $p_j \in \cX$ with $x^*_j>0$ satisfying  $p_j\in \cX_i$ and $y^*_j$ value not yet set, we set: \[y_j^*\leftarrow  \left (\sum_{ p_\ell\in \rball{p_j}{\gamma/3} \cap \cX_i} x_\ell^* \right ) \textup{and } y_\ell^*\leftarrow 0 \textup{ for all }p_\ell\in \rball{p_j}{\gamma/3} \cap \left(\cX_i \setminus \{ p_j\} \right).\]
    \item[(b)] Finally, for all $p_j \in \cX$ with $x^*_j=0$, we set $y^*_j\leftarrow 0$.
\end{itemize}

Informally, we are just moving  weight to $p_j$ from points of the same group (as $p_j$) that are at a distance strictly less than $\gamma/3$ from $p_j$.
\begin{lemma}
$\{y_j^*\}_{j\in [n]}$ satisfies Constraints (1'-4').
\end{lemma}
\begin{proof}
Observe that $\{y_j^*\}_{j\in [n]}$ satisfies the constraint (\ref{eq:constraint4new}).  If a point $p_j \in \cX_i$ satisfies $y^*_j>0$, then, it means that we set $y^*_\ell$ to $0$ for every $p_l \in \rball{p_j}{\gamma/3} \cap \left(\cX_i \setminus \{ p_j\} \right)$.

Constraint (\ref{eq:constraint2new}) is  satisfied because
\[
\sum_{p_\ell \in \rball{p_j}{\gamma/6}} y_\ell^*
\leq 
\sum_{p_\ell \in \rball{p_j}{\gamma/6+\gamma/3}} x_\ell^*= 
\sum_{p_\ell \in \rball{p_j}{\gamma/2}} x_\ell^*\leq 1
\]
since $\{x_\ell^*\}_{\ell\in [n]}$ satisfies constraint (\ref{eq:constraint2}). Constraint (\ref{eq:constraint1new}) is satisfied because $\sum_{p_j \in \mathcal \cX_i} y^*_j= \sum_{p_j \in \mathcal \cX_i} x^*_j$ and  Constraint (\ref{eq:constraint3new}) is trivially satisfied.
\end{proof}

We next pick a random permutation $\sigma$ as in the previous Section~\ref{subsec:randomLP}, but now using the values $\{y_\ell^*\}_{\ell\in [n]}$. We add $p_j$ to the output $\cS$ if $\sigma(j)\leq \sigma(\ell)$ for all $p_\ell$ such that $d(p_\ell,p_j)<\gamma/6.$ Note that all points in $\cS$ are therefore at least a distance of $\gamma/6$ apart. 

\begin{theorem}\label{thm:6approx}
Asssume $k_i\geq 3\epsilon^{-2} \log (2m)$ for all $i\in [m]$. There is a $\poly(n,k, \delta^{-1})$ time algorithm that returns a subset of points with diversity $\lsf/6$ and includes $(1-\epsilon)k_i$ points in each group $i\in [m]$ with probability at least $1-\delta$.
\end{theorem}

\begin{proof}
Let $Y_p=1$ if the point $p\in \cX$ is included in the output $\cS$. Fix $i\in [m]$. The proof of Lemma \ref{lem:randround} applied to balls of radius $\gamma/6$ rather than balls of radius $\gamma/2$, ensures that for each $i\in [m]$,
$
\E[\sum_{p\in \cX_i} Y_p] \geq k_i$.
The fact  $\{Y_p\}_{p \in \cX_i}$ are fully independent allows us to apply the Chernoff bound and conclude
$
\Pr  [\sum_{p\in \cX_i} Y_p\leq (1-\epsilon) k_i   ] \leq \exp(-\epsilon^{2} k_i/3)\leq 1/(2m)$. 
Hence, by an application of the union bound, we ensure that with probability at least $1/2$, $|\cS \cap \cX_i| \geq (1-\epsilon)k_i$ for all $i\in [m]$. Repeating the process $\log \delta^{-1}$ times ensures that at least one of the trials succeeds with probability at least $1-\delta.$
\end{proof}

\changes{Note that Theorem~\ref{thm:6approx} requires the $k_i$ values to be sufficiently large, and such conventions have also been used in prior work~\cite{Bhaskara16}. For small $k_i$ values, i.e., $k_i =o (\log n)$, the \fGMM algorithm introduced in Moumoulidou et al.~\cite{moumoulidou21} obtains a 5-approximation guarantee in polynomial time. Using an additive Chernoff bound, alternatively, we can find at least $k_i - O(\sqrt{k_i \log m})$ points from each group $i \in [m]$, without the requirement of having large $k_i$'s.}

\subsection{$(m+1)$-Approx with Perfect Fairness}\label{subsec:lineartime}

We now describe {\genswap}, an  $m+1$-approximation algorithm that ensures perfect fairness. \changes{This is an improvement over the previously known $3m-1$ approximation~\cite{moumoulidou21}. Specifically, we give the pseudocode for \genswap~(Algorithm~\ref{alg:genswap}) and show that \genswap~returns a solution which is an $m+1$-approximation for \fMM. We also present an example in Appendix~\ref{app:tightexample} that shows that the analysis presented in Moumoulidou et al.~\cite{moumoulidou21} cannot be improved to obtain a better approximation factor. Missing details from this section are presented in Appendix~\ref{app:lineartime}.}

\smallskip

\noindent 
\textbf{Overview of \genswap}. 
We assume a guess $\gamma$ for $\lsf$. The algorithm proceeds by iteratively
building clusters of close points of distinct groups. Our main idea is to
select one point from each cluster such that the fairness constraints are
guaranteed. First, we describe the procedure for building a cluster. Let $D$
denote a cluster initialized with a point of group $i \in [m]$. Among the
available points $\cR$, we include a point $p \in \cR$, if it is within a
distance of $\frac{\gamma}{m+1}$ to some point $x \in D$, and no other point of
the same group is already present in $D$.

If there is no such point, the cluster $D$ is complete, and we remove all
points from $\cR$ that are within a distance of $\frac{\gamma}{m+1}$ from some
point in $D$. Also, we discard all points of group $i$, i.e., $\cX_i$ from
$\cR$, as soon as there are at least $k$ distinct clusters in $\mathcal C$
containing points from $\cX_i$. We continue this process of iteratively
building clusters, until there are points from each group that are part of at
least $k$ distinct clusters or if there are no remaining points.

Next, we use an approach similar to~\cite{moumoulidou21} and select at most one
point from each cluster, satisfying the fairness constraints. We construct a
flow network with clusters $D_1, D_2, \cdots, D_t$ in $\mathcal C$ represented
by nodes $v_1, v_2, \cdots, v_t$ and groups represented by nodes $u_1, u_2,
\cdots, u_m$. We add an edge with capacity $1$ between every pair $u_i$ and
$v_j$ if there is a point of group $i$ in cluster $D_j$ for some $j \in [t]$.
We create a source node $a$ and add edges with capacity $k_i$ between $a$ and
$u_i \ \forall i \in [m]$. We then create a sink node $b$ and add edges with
capacity $1$ between $b$ and $v_j \ \forall j \in [t]$. Finally, we
find maximum flow using Ford-Fulkerson algorithm~\cite{cormen2009introduction}.
For each edge $(u_i, v_j)$ with flow equal to $1$, we include the point of
group $i$ from cluster $D_j$ in our solution. \changes{We conclude this section with the following theorem:}

\begin{theorem}\label{thm:lineartime}
\genswap~Algorithm returns an $(m+1) (1+\epsilon)$-approximation and achieves {perfect fairness} for the \fMM problem using a running time of $O(nkm^3 \epsilon^{-1}\log n)$.
\end{theorem}

\subsection{Hardness of Approximation}\label{subsec:hardness}
In this section, we give a hardness of approximation result for the \fMM problem. Our result is a generalization and improvement over the $2$-approximation hardness shown in~\cite{moumoulidou21}, as we also allow for approximations in fairness constraints.

\begin{definition}[\textsc{Gap-Clique}$_\rho$] Given a constant $\rho \geq 1$, a graph $G$, and an integer $k$, we want to distinguish between  the case where a clique exists of size $k$ (the ``yes" case) and the case where no clique exists of size $\geq k/\rho$ (the ``no" case).
\end{definition}

It is known that $\textsc{Gap-Clique}_\rho$ is NP-hard for every $\rho \geq
1$~\cite{pcp}. Now, via a reduction from the $\textsc{Gap-Clique}_\rho$ we argue that \fMM cannot be approximated to a factor better than $2$, even allowing for multiplicative approximations in fairness constraints. 
See Appendix~\ref{app:genmetrics} for the proof of the following theorem.

\begin{theorem} \label{thm:hardness}
Let  $\alpha < 2$ and $\beta > 0$ be constants.
Unless $P = NP$, there is no {polynomial time} algorithm for the \fMM problem that obtains an $\alpha$-approximation factor for diversity score, and
$\beta$ fairness.
\end{theorem}

\section{Euclidean Metrics}\label{sec:eucmetrics}

In this section, we assume that the metric space is Euclidean, i.e., we can
associate a point $p_i\in {\mathbb R}^D$ with the $i$th entry of $\cX$ and
$d(p_i,p_j)=\|p_i-p_j\|_2=\sqrt{\sum_{\ell\in [D]} (p_i(\ell)-p_j(\ell))^2} \
.$ When $D=1$ we show that the problem can be solved exactly in polynomial time
via Dynamic Programming. More generally, when $D=O(1)$ we present a bi-criteria
approximation that uses an extension of the dynamic programming approach and
properties of low dimensional Euclidean spaces. Missing details from this
section are presented in Appendix~\ref{app:eucmetrics}.

\subsection{Exact Computation in One Dimension}\label{subsec:line}
In this section, we assume the points in the universe $\cX =\bigcup_{i=1}^{m}
\cX_i$ can be embedded on a line. Specifically, let $\cX=\{p_1, \ldots, p_n\}$
where each $p_i\in \mathbb{R}$ and we order the points such that $p_1\leq p_2
\leq \ldots \leq p_n$. We further assume a guess $\gamma$ on the optimal
diversity score for \fMM and design the dynamic programming algorithm
\fL~(Algorithm~\ref{alg:dp1}) that computes an exact solution when $\gamma =
\lsf$. See the previous section for a discussion on guessing $\gamma$.

\medskip
\noindent \textbf{Dynamic Programming.}
Define the dynamic programming table $\tb \in \{0, 1\}^{(k_1+1)\times \ldots
\times (k_m+1) \times n}$ indexed from 0. An entry $\tb[k'_1, k'_2, \cdots,
k'_m, j] \in \{0, 1\}$ is 1 iff there is a subset $\cS^{'}$ of the first $j$
points on the line with diversity $\gamma$ that contains $k'_i$ points from
each group $i\in [m]$. To compute the entries of $\tb$, we process the points
in their order of appearance on the line.

Note that there is a set $\cS^{'}$ with $k'_i$ points from each group $i$ among
the first $j$ points if: (1)~there is such a set among the first $j-1$ points,
or (2)~point $j$ belongs to group $i$ for some $i \in [m]$, and among the first
$j'$ points there is a set with $ k'_1, \cdots, k'_i-1, \cdots, k'_m $ points
from the corresponding groups where $j'<j$ is the largest value such that
$d(p_j,p_j') \geq \gamma$.

See \fL (Algorithm \ref{alg:dp1}) for the resulting algorithm. For simplicity,
the algorithm is written to only determine whether it is possible to pick a
subset with diversity $\gamma$ subject to the required fairness constraints.
However, the algorithm can be easily extended to construct a subset of points
for every non-zero entry in $\tb$ by storing a pointer to the choice we made.
For an entry $\tb[k'_1, k'_2, \cdots, k'_m, j]=1$ that also satisfies
$\tb[k'_1, k'_2, \cdots, k'_m, j-1]=1$ we store a pointer to that entry.
In the second case, if $\tb[k'_1, k'_2, \cdots, k'_m, j']=1$ for some $j'$, we
store a pointer to that entry. We construct the solution set using the stored
pointers, starting at $\tb[k_1, k_2, \cdots k_m, n]$ and backtracking, to
indicate which points to add to the solution.


\begin{algorithm}[t] 
	\caption{\fL: An exact algorithm for data on a line}\label{alg:dp1} 
		{\small
			\begin{algorithmic}[1]
				\Statex
				\begin{description} 
					\item[\rlap{Input:}\phantom{Output:}]$\cX = \bigcup_{i=1}^{m} \cX_i$: Universe of available points.
					\item[\phantom{Output:}] $k_1, \ldots, k_m \in \mathbb{Z}^{+}$.
					\item[\phantom{Output:}] $\gamma \in \mathbb{R}^{+}$: A guess of the optimum fair diversity.
					\item[Output:] $k_i$ points in $\cX_i$ for $i\in [m]$.
				\end{description} 
				
				\State Let $n \leftarrow |\bigcup_{i =1}^m \cX_i |$ and initialize  $\tb \in \{0,1\}^{(k_1+1)\times \ldots (k_m+1) \times n}$ to 0.
				\State Set $\tb[0,\ldots ,0  ,0]\leftarrow 1$, $\tb[0,\ldots ,0  ,1]\leftarrow 1$, and if $p_1\in \cX_\ell$, $\tb[0,\ldots, \underbrace{1}_{\textup{index $\ell$}},\ldots ,0  ,1]\leftarrow 1.$	
			\For{$j=2$ to $n$} 
				\State Let $i \in [m]$ satisfy $p_j\in \cX_i$. \State Let $j'=\max \left( \{0\} \cup \{j'\in [n]:p_{j'}+\gamma\leq  p_j\} \right)$.
				\For{$k_1'\in \{0,\ldots, k_1\},\ldots, k_m'\in \{0,\ldots, k_m\}$}
				\State 
				$\tb[k'_1, \cdots, k'_m, j]\leftarrow \tb[k'_1,  \cdots, k'_m, j-1]$.
				\State If $k'_i\geq 1$, 
				$\tb[k'_1, \cdots, k'_m, j]\leftarrow \tb[k'_1,  \cdots, k'_i-1, \ldots, k'_m, j']\vee \tb[k'_1,  \cdots, k'_m, j-1].$
				\EndFor
			\EndFor
					 \State \Return $\tb[k_1, k_2, \cdots, k_m, n]$.
			\end{algorithmic}
			}
\end{algorithm}
\setlength{\textfloatsep}{2pt}

\begin{theorem}\label{thm:fairline}
There is an algorithm that solves the \fMM problem exactly when the points can be embedded on a line and requires a running time of $O(n^4 \prod_{i=1}^m (k_i+1))$. 
\end{theorem} 

\subsection{Coresets for Constant Dimensions}\label{subsec:euclCoresets}
In this section, we design efficient $(1+\epsilon)$-coresets for \fMM in metric
spaces of low doubling dimension (Definition~\ref{def:doubling}). Let $\lambda$
denote the doubling dimension of $\cX$. Our approach generalizes prior work on constructing efficient coresets for unconstrained Max-Min diversification~\cite{Ceccarello17} to \changes{the} \fMM problem.

Specifically, we give the first algorithm for constructing coresets in metric
spaces of doubling dimension. The proposed approach uses the \gonzalez
algorithm that obtains a factor $2$-approximation for the unconstrained Max-Min
diversification problem~\cite{ravi1994,Tamir}.

\gonzalez is a greedy algorithm and works as follows: it starts with an
arbitrary point in a set $S$ and in every subsequent step selects the point that is the
farthest away from the previously selected points. In fact, readers familiar
with the $k$-center clustering problem will recognize that this is the same
strategy used by~\cite{Gonzalez1985ClusteringTM}. If $k$ is the size of the
subset to be selected and $n$ is the size of the universe of points, it is
known that~\gonzalez can be implemented in $O(kn)$ time~\cite{pmlr-v97-kleindessner19a, Wang2018}.

\medskip

\noindent
\textbf{Coreset Construction.} 
First, define  $\foo = \appr / (1+ \appr)$ and note that $\appr/2 \leq \foo <
1$ since $\appr \in (0, 1]$. The \textsc{Coreset} Algorithm constructs coreset $\cT$ as follows: we run
\gonzalez on each group $i \in [m]$ separately to retrieve a set $T_i$ with
$O((4/\foo)^\doubldim k)$ points. The coreset $\cT$ is equal to the union of
the $T_i$ sets for all $i \in [m]$, namely: $\cT \leftarrow \bigcup_{i=1}^{m} T_i, \ \textup{where} \ T_i \leftarrow \gonzalez(\cX_i, (4/\foo)^\doubldim k).$

We will show that $\cT$ contains a set $\cT^{'}$ with $\diver(\cT')\geq \lsf/(1+\appr)$ and $k_i$
points from each group $i$. At a high level, the idea is that for each group $i$ there are two
cases: (1)~either $T_i$ contains a sufficient number of points that are far
apart such that even if we had to remove points close to points
selected from other groups, we would still have enough points to satisfy
fairness, or (2)~the optimal points from group $i$ are within small distance
from their closest point in $T_i$. In the analysis we show that in both cases
we have enough points from each group $i$ to satisfy fairness while these
points are at least $\lsf/(1+\epsilon)$ far apart.

We first prove the following lemma, which we will use later. 

\begin{lemma}\label{lem:doublingDim}
Let $S$ be a set of $k^{'}=(4/\foo)^\doubldim k$ points that are all at least
$(\foo/2) \gamma$ far apart. Then, there exists a subset $S^{'} \subset S$ of points that are all at least $\gamma$ far apart and $|S^{'}| \geq k$. 
\end{lemma}

\begin{proof}
Let $S'=\emptyset$. Add an arbitrary point $x$ from $S$ to $S'$ and remove all points in the ball $\rball{x}{\gamma}$ from $S$. Consider a set of balls of radius $(\foo/4) \gamma$ that cover the removed points. Each of these balls cover at most one removed point since discarded points are at least $(\foo/2) \gamma$ far apart. Hence, the number of balls is at least the number of removed points. But because the doubling dimension is $\lambda$ we know there exists a set of $(4/\foo)^\lambda$ balls of radius $(\foo/4) \gamma$ that cover the removed points. Hence, the number of removed points is at most $(4/\foo)^\doubldim$. Since there were $k^{'}=(4/\foo)^\doubldim k$ points in $S$, we may continue in this way until we've added $k$ points to $S'$. All chosen points are at least $\gamma$ apart as required.
\end{proof}

Our main theorem in this section is as follows:

\begin{theorem}\label{thm:coresets_doubling}
There is an algorithm that returns a $(1+\epsilon)$-coreset of size $O((8/\epsilon)^\doubldim km)$ in metrics of doubling dimension $\doubldim$ with a running time  $O((8/\epsilon)^\doubldim kmn)$. 
\end{theorem}

\begin{proof}
We show that the set $\bigcup_{i=1}^{m}T_i$ constructed by the
\textsc{Coreset} Algorithm is an $(1+\epsilon)$-coreset by showing the existence of a set $\cT^{'} \subseteq \bigcup_{i=1}^{m} T^{}_i$ with $k_i$ points from each group $i$ and $\diver{(\cT^{'})} \geq \lsf/
(1+\epsilon)$. 

For every group $i \in [m]$, we define $\widehat{T}_i$ to be the
maximal prefix of the points added by  \gonzalez to form $T_i$ such $\diver(\widehat{T}_i)\geq (\foo/2) \divScore$. We first process all the groups for which $|\widehat{T}_i| <
(4/\foo)^\doubldim k$, which we call \emph{critical} groups. For all critical
groups, any point $p \in \cX_i \setminus \widehat{T}_i$ is  within distance
$(\foo/2) \divScore$  from its closest point $f(p)$ in $\widehat{T}_i$,
i.e., $d(p, f(p)) < (\foo/2) \divScore$. As a result, for any pair of optimal
points $o_1, o_2$ in critical groups we deduce:

\begin{eqnarray*}
	d(f(o_1), f(o_2)) & \geq & d(o_1, o_2) - d(o_1, f(o_1)) - d(o_2, f(o_2)) \\
	&>&  \divScore - {2 \cdot \foo \divScore}/{2}  ={\divScore }/(1+\appr) \ .
\end{eqnarray*}

\looseness -1
We initialize $\cT^{'} = \bigcup_{o \in \cup_{i:\textup{critical}}
\cS^{*}_i}f(o)$ where $\cS^{*}_i$ is the set of points in an optimal solution
belonging to group $\cX_i$. We now process all \emph{non-critical} groups $j
\in [m]$ in an arbitrary order and remove any point in $\widehat{T}_j$ that is
less than $\divScore$ apart from some point in $\cT^{'}$. Then we argue that in
the remaining points there is a set of points $T^{'}_j$ with $k_j$ points that
are at least $\lsf$ far apart.

\sloppy
By the doubling dimension property and the fact that all the points in $\widehat{T}_j$ are at least $(\foo/2) \divScore$ far apart, the removal step described above discards at most $(4/\foo)^\doubldim \sum_{i:\small{\textup{processed \
groups}}}|\cT^{'} \cap \cX_i|$ points from $\widehat{T}_j$. Consequently, regardless of the order in which we process the \emph{non-critical} groups, by the time we process $\widehat{T}_j$ for some $j \in [m]$, there will be \emph{at least} $(4/\foo)^\doubldim k - \sum_{i:\small{\textup{processed \ groups}}} (4/\foo)^\doubldim k_i \geq (4/\foo)^\doubldim k_j$ points that are at least $(\foo/2) \divScore$ apart from each other.

Now by applying Lemma~\ref{lem:doublingDim} on the points of $T^{'}_j$, we conclude that there are at least $k_j$ points within $\divScore$ distance from all other points in $\cT^{'}$. Then this set of points $T^{'}_j$ can be added to $\cT^{'}$ to satisfy fairness for group $j$. Thus, it holds that $\diver{(\cT^{'})} \geq \divScore / (1+\epsilon)$ which implies the claimed approximation factor for coreset $\cT$. 

As $\foo = \epsilon/(1+\epsilon) \geq \epsilon/2$, we have $|\cT| = O((8/\epsilon)^\doubldim km)$. Since we use \gonzalez to obtain $\cT$, the running time of the \textsc{Coreset} algorithm is $O((8/\epsilon)^\doubldim kmn)$.
\end{proof}

{From the coreset $\cT$, we can obtain a $(1+\epsilon)$-approximation by enumerating over all subsets of $\cT$ and returning the subset with maximum diversity and perfect fairness. The running time of this algorithm is $O(2^{O(k)} + nk)$, when $m, \lambda$ are constants. In the next section, we describe an algorithm that has a polynomial dependence on $n$ and $k$, obtained at the cost of $(1-\epsilon)$-fairness.}

\subsection{$(1+\epsilon)$ Approx with $(1-\epsilon)$ Fairness}\label{subsec:bicriteria}
In this section, we describe \fB (Algorithm~\ref{algo:m4}) which uses $(1+\epsilon)$-coresets described in Section~\ref{subsec:euclCoresets} and returns a subset of points with diversity at least $\lsf/(1+\epsilon)$ and has $(1-\epsilon)k_i$ points from each group $i \in [m]$. Missing details are presented in Appendix~\ref{app:eucmetrics}.

First, we discuss $\fDP$ (Algorithm~\ref{alg:dp2}), which is a dynamic
programming subroutine used in \fB. The subroutine will be applied to a collection of $t$ disjoint subsets of $\cX$: $\cC = \{C_1, C_2, \dots, C_t \}.$ This collection will be \emph{well-separated} in the sense that for all $i\neq j$ and $x \in C_i, y \in C_j$ then $d(x, y) \geq \gamma$. Points in the same set can be arbitrarily close together. We design \fDP (Algorithm~\ref{alg:dp2}): a dynamic programming algorithm to retrieve a set $\cF = \bigcup_{i=1}^{m} \cF_i \subseteq \cC$ with $k_i$ points per group $i$ and $\diver{(\cF)} \geq \gamma$ if such a set exists in $\cC$.

\medskip
\noindent \textbf{Dynamic Programming.} Define the dynamic programming table $\tb \in \{0, 1\}^{(k_1+1)\times \ldots \times (k_m+1) \times t}$ indexed from 0. An entry $\tb[k'_1, k'_2, \cdots, k'_m, j] \in
\{0, 1\}$ is 1 iff there is a subset $\cF'$ among the first $j$ clusters such that $|\cF' \cap \cX_i| \geq k'_i \ \forall i \in [m]$ and $\diver({\cF'}) \geq \gamma$.

To compute the entries of $\tb$, we process the clusters in $\cC$ using some fixed ordering. Note that there is a set $\cF'$ with $k'_i$ points from each group $i$ among the first $j$ clusters if there is a subset $P \subseteq C_j$ with $\diver{(P)} \geq \gamma$ and $p'_i$ points from each group $i$; and, among the first $j-1$ clusters, there is a set with $ k'_1-p'_1, k'_2 - p'_2, \cdots k'_i - p'_i \cdots,k'_m - p'_m$ points from each group $i \in [m]$ that are at least $\gamma$ far apart (the function $f$ in \fDP (Algorithm \ref{alg:dp2}) evaluates where there is such a set $P$). We enumerate over all possible subsets of $C_j$ to identify the subset $P$. 

See \fDP (Algorithm \ref{alg:dp2}) for additional details and implementation. For simplicity, the algorithm is written to only determine whether it is possible to pick a subset with diversity $\gamma$ subject to the required fairness constraints. Similar to \fL, the algorithm can be easily extended to construct a subset of points for every non-zero entry in $\tb$ by storing a pointer to the choice we made. 

\begin{algorithm}[t!]
	\caption{\fDP: A dynamic programming subroutine}\label{alg:dp2} 
	{\small
	\begin{algorithmic}[1]
		\Statex
		\begin{description}
			\item[\rlap{Input:}\phantom{Output:}] $C_1, C_2,\cdots, C_t$: Family of disjoint subsets of  $\cX = \bigcup_{i=1}^{m} \cX_i$.
			\item[\phantom{Output:}] $k_1, \ldots, k_m \in \mathbb{Z}^{+}$.
			\item[\phantom{Output:}] $\gamma \in \mathbb{R}^{+}$: A guess of the optimum fair diversity. 
			\item[Output:] $k_i$ points in $\cX_i$ for $i\in [m]$. 
		\end{description} 
		\State Define boolean function $f(p'_1, \ldots, p'_m,j)$ that evaluates to 1 iff there exists $P\subseteq C_j$ with $\diver(P)\geq \gamma$ and  $|P\cap \cX_i|=p'_i$ for all $i\in [n]$.
			\State  Initialize $\tb \in \{0,1\}^{(k_1+1)\times \ldots (k_m+1) \times t}$ to 0.
			\State Set $\tb[p'_1, \ldots p'_i,\ldots, p'_m,1]\leftarrow f(p_1', \ldots, p_m',1)$.
			\For{$j=1$ to $t$}
			\State For $k'_i \in \{0, \cdots, k_i\} \ \forall i \in [m]$, update the entries in $\tb$ as:
			
			\begin{equation*}
				 \tb[k'_1,  \cdots, k'_m, j] \leftarrow 
				    \bigvee\limits_{\substack{p'_i \leq k'_i \\ \forall i \in [m]}} \tb[k'_1-p'_1, \cdots, k'_m-p'_m, j-1]f(p'_1, \ldots, p'_m,j).
			\end{equation*}
			\EndFor
			\State \Return $\tb[k_1, k_2, \cdots, k_m , n]$. 
		\end{algorithmic}
	}	
\end{algorithm}

\begin{theorem}\label{thm:fairDP}
If $\gamma = \lsf$, then, \fDP (Algorithm~\ref{alg:dp2})
returns a set $\cS$ that satisfies $\diver(\cS) \geq \lsf$ and $|\cS \cap \cX_i| \geq k_i \ \forall i \in [m]$ and has a running time of $O(\prod_{i=1}^{m}(k_i +1)^{2} 2^R t)$ where $R = \max\{|C_1|, |C_2|, \cdots, |C_t|\}$. 
\end{theorem}

\begin{algorithm}[t!]
	\caption{\fB: A bi-criteria algorithm}\label{algo:m4} 
	{\small
		\begin{algorithmic}[1]
			\Statex
			\begin{description} 
				\item[\rlap{Input:}\phantom{Output:}] $\cX = \bigcup\limits_{i=1}^{m} \cX_i$: points in $\mathbb{R}^D$ with doubling dimension $\lambda$.
				\item[\phantom{Output:}] $k_1, \ldots, k_m \in \mathbb{Z}^{+}$.
				\changes{\item[\phantom{Output:}] $\cT = \bigcup_{i=1}^{m} T_i$: A coreset for fair Max-Min.}
				\item[\phantom{Output:}] $\gamma\in \mathbb{R}^{+}$: A guess of the optimum fair diversity.
				\item[\phantom{Output:}] $\epsilon \in [0, 1]$: approximation error parameter.
				\item[Output:] $k_i$ points in $\cX_i$ for $i\in [m]$.
			\end{description}
			\State $\widehat{T}_i \gets$ a maximal prefix of points in $T_i$ such that $\diver{(\widehat{T}_i)} \geq \epsilon\gamma/4$.
			\State $p \gets $ a point selected uniformly at random from $[0, W]^\dimension$, where $W = 2m\dimension\gamma/\epsilon$.
			\State Construct axis-aligned cubes $\mathcal{C} =\{ C_1, C_2, \cdots, C_t \}$ of side length $W$ using $p$ as one of the corners.
			\State In each cube $C_i$, remove all the points that are within a distance of $\gamma/2$ from one of the boundaries.
			\State  \Return $\cS \leftarrow \fDP(C_1, \cdots , C_t, (1-\epsilon)k_1, \cdots (1-\epsilon)k_m, \gamma)$.
		\end{algorithmic}
	} 
\end{algorithm} 

\smallskip

Now, we describe a $1+\epsilon$ approximation algorithm for Euclidean metrics called \fB that achieves $1-\epsilon$ fairness.

\smallskip
\noindent
\textbf{Overview of \fB.}
\changes{As part of the input, we construct a $(1+\epsilon)$-coreset $\cT = \bigcup_{i=1}^{m} T_i $ of size $O((8/\epsilon)^\doubldim km)$ using the  \textsc{Coreset} algorithm described in Section~\ref{subsec:euclCoresets}. We further assume a guess $\gamma$ for the optimal diversity score $\lsf$. Note that the coreset $\cT$ is only constructed once and used for different guesses of $\lsf$}.

\changes{For a fixed guess $\gamma$}, for every group $i \in [m]$, we select a maximal \changes{prefix} of points $\widehat{T}_i \subset T_i$ that are at least $\epsilon \gamma/4$ far apart and define $\widehat{\cT} = \bigcup_{i=1}^{m} \widehat{T}_i$.

Our main idea is to partition $\widehat{\cT}$ and obtain a collection of sets $\cC = \{ C_1, C_2, \cdots, C_t \}$ separated by at least $\gamma$ distance; thus any pair of points $x \in C_i \ \textup{and} \ y \in C_j, \ \forall i,j \ \textup{such that} \ i \neq j$, is separated by distance at least $\gamma$. Then, we use \fDP on these sets $C_1, C_2, \cdots, C_t$, and recover a solution $\cS$ with diversity $\gamma$.

To this end, we partition the points in $\widehat{\cT}$ into axis-aligned \emph{cubes} $\cC = \{ C_1, C_2, \cdots, C_t\}$ of length
$W=2m\dimension\gamma/\epsilon$ as follows: we select a point $p$ uniformly at random from $[0, W]^D$. Using $p$ as one of the corners, we form axis-aligned
cubes of length $W$ until every point in $\cX$ is in one of the cubes. Then, from every cube $C_i \ \forall i \in [t]$ we remove every point of $\widehat{T}$ that is within a distance of $\gamma/2$ from one of its boundaries. Notice that any point that was not removed from a cube is at least $\gamma$ far apart from any other point in a different cube. However, points within the same cube can be arbitrarily close. It is now easy to see that we can use \fDP (Algorithm~\ref{alg:dp2}) on $\cC$ to retrieve a sufficient number of points from each group in $[m]$.

In the analysis below, we show that with probability at least $1/2$, we are able to find at least $(1-\epsilon)k_i$ points from each group $i \in [m]$ that are all $\gamma$ far apart. 
 
\smallskip

\noindent
\textbf{Analysis.} Let $\cS^{*}=\bigcup_{i=1}^{m}\cS^{*}_i \subset \cT$ denote
the optimal solution for \fMM on the coreset ${\cT} = \bigcup_{i=1}^{m} {T}_i$ with $\diver{(\cS^{*})} \geq \lsf/(1+\epsilon)$. Note that the optimal solution in $\cT$ is some subset in $\widehat{\cT}$ (see Theorem~\ref{thm:coresets_doubling}).

As a first step, we bound the number of optimal points $\cS^*_i$ from a group $i \in [m]$ that are removed by \fB because they are within a distance of $\gamma/2$ from one of the boundaries of a cube. 

\begin{lemma}\label{lem:numpoints}
$\Pr[\forall i \in [m] \ : |\bigcup_{j \in [t]} C_j
\cap \cS^{*}_i| \geq (1-\epsilon) k_i] \geq 1/2$. \end{lemma}
\begin{proof}
Let $T'_i = \bigcup_{j \in [t]} C_j \cap \widehat{T}_i$ be the remaining points in $\widehat{T}_i$ that are not close to the boundaries of any cube. Note that the \fB algorithm \emph{succeeds} if after the removal step there are least $(1-\epsilon)k_i$ optimal points from each group $i$ that can be selected by \fDP at the final step of the algorithm while it \emph{fails} otherwise. Below, we show that the probability it \emph{succeeds} is at least $1/2$. 
 		
We compute the probability that a point $q \in \widehat{T}_i$ is not removed by \fB, i.e., $q \in T'_i$. It is removed if it lies within a distance of $\gamma/2$ from its boundaries in each dimension. Therefore, for $q$ to remain in $T'_i$, the point $p$ selected randomly from $[0, W]^D$ must not fall within a range of total length $\gamma$, in each dimension, which gives us:

\[
\Pr[q \not\in T'_i] = 1- \Pr[q \in T'_i]  = 1- \left(\frac{W-\gamma}{W}\right)^\dimension  \leq  \gamma \dimension/{W} = \epsilon/2m  \ .
\]

Fix a specific optimum solution. Define $A_i$  be the number of points removed from this solution that are in group $i$. By Markov's inequality,
$\Pr[A_ i \geq k_i \epsilon] \leq \frac{\E[A_i]}{k_i \epsilon} \leq 
\frac{{k_i \epsilon}/{(2m)}}{k_i \epsilon}= \frac{1}{2m}.$

Taking union bound over all groups $i \in [m]$, we can bound the
probability of discarding more than $k_i\epsilon$ points from some group $i$,
$\Pr[\exists i \in [m] \ : \ A_i \geq k_i \epsilon] \leq \sum_{i=1}^{m} \Pr[A_ i \geq k_i \epsilon] < {1}/{2} ~,$
and the lemma follows.
\end{proof}

\fDP depends exponentially on the number of points remaining in each cube (see Theorem~\ref{thm:fairDP}). Now, we show that the total number of points remaining in each cube does not depend on $n$ or $k$, and depends only on $m, D, \epsilon$.
\begin{lemma}\label{lem:packing}
$|C_j| \leq m \cdot ( 8mD^{3/2}/\epsilon^2 )^\lambda$ for all $j \in [t]$.
\end{lemma}

We showed that for a fixed guess $\gamma$, the \emph{success} probability of \fB is $\geq 1/2$. Note that the only randomization used by \fB is in selecting $p$. In order to increase the probability of success to $1-\delta$ for some small $\delta \in (0, 1)$, we repeatedly select $\eta$ points uniformly at random from $[0, W]^D$ as the corners. For each corner, we obtain a solution using \fB, and we output the solution with the biggest diversity which also satisfies the fairness constraints with a loss of $(1-\epsilon)$ multiplicative factor. The value of $\eta =\log(1/\delta)$ is selected such that the failure probability is $(1/2)^\eta < \delta$.

\changes{Note that the construction of the coreset $\cT$ allows us to reduce the number of guesses on $\lsf$ from $O(n^2)$ to $O(|\cT|^2) = O((8/\epsilon)^{2\doubldim} k^2 m^2)$, which are all the pairwise distances in $\cT$.} Further, the number of clusters (i.e., cubes) in \fB is upper bounded by the size of the coreset $\cT$, \changes{which does not depend on $n$}. The running time of \fB depends on the running time to construct the coreset, which is $O((8/\epsilon)^\doubldim kmn)$, and the running time of \fDP (Algorithm~\ref{alg:dp2}) on the cubes $\cC$. Since the number of points in each cube is $O(m \cdot (8mD^{3/2}/\epsilon^2 )^\lambda )$, we conclude with the following theorem:

\begin{theorem}\label{thm:euclidean}
If $\gamma \geq \lsf/(1+\epsilon)$, \fB Algorithm returns a set $\cS$ such that $\diver{(\cS)}\geq \lsf/(1+\epsilon)$ and $|\cS \cap \cX_i| \geq  k_i(1-\epsilon) \ \forall i \in [m]$ with probability at least $1-\delta$. 
For constant $D, m$, the running time is $O( nk + \poly(1/\epsilon,k,\log(1/\delta)))$.
\end{theorem}

\changes{In Appendix~\ref{app:eucmetrics}, we give the exact running time with all the parameters. We can observe that the running time  depends doubly exponentially on the doubling dimension, which is not uncommon for diversity maximization in doubling dimension metrics~\cite{Ceccarello2018, Ceccarello17}.}

\section{Scalable Implementations}\label{sec:scalable}

\subsection{Data Stream Algorithms}\label{subsec:streams}
In this section, we present single pass data stream algorithms that obtain the same approximation guarantees as that of sequential algorithms, while using low space. Missing details from this section are presented in Appendix~\ref{app:streaming}. 

\subsubsection{Extending Previous Algorithms} First, we describe an algorithm called \tauGMM that processes points sequentially, and includes a point in the solution if it is at least the threshold $\tau$ apart from every point in the current solution set. The set of points returned by \tauGMM are all separated by a distance of at least $\tau$. If $m = 1$, then, we can set $\tau = \lsf/2$ (using guessing for $\lsf$), and \tauGMM returns a solution set that is also a $2$-approximation for the \fMM problem~\cite{CormodeMZ07}. \tauGMM allows us to extend it to data streaming setting, unlike the \GMM algorithm which requires identifying the maximum distance point in each iteration. 

Using \tauGMM with $\tau = \lsf/2$, we can obtain a $5$-coreset for general metrics~\cite{moumoulidou21}, and $(1+\epsilon)$-coreset for Euclidean metrics (Section~\ref{subsec:euclCoresets}). Then, on the coreset, we use the randomized rounding algorithm from Section~\ref{subsec:randomLP2} and return the solution. This approach gives us the following guarantees:

\begin{corollary}\label{thm:streameuc_informal}
There is a $O(\epsilon^{-1} km\log n)$-space data stream algorithm that returns a $30(1+\epsilon)$-approximation with $(1-\epsilon)$-fairness for general metrics. For Euclidean metrics, there is a $O((8/\epsilon)^{\lambda}km\epsilon^{-1} \log n)$ space data stream algorithm that returns a $1+\epsilon$-approximation with $(1-\epsilon)$-fairness where $\lambda$ is the doubling dimension of $\cX \subset \mathbb{R}^D$.
\end{corollary}

\subsubsection{Improved Result for $m = 2$} In~\cite{moumoulidou21}, the authors describe an algorithm called \textsc{Fair-Swap} which returns a $4$-approximation to the \fMM problem when the number of groups is $m = 2$. The algorithm can be directly extended to a $2$-pass streaming algorithm using $O(k)$ space with the same $4$-approximation guarantee. Building upon their work, and using new ideas we obtain a \emph{single pass} algorithm \streamrbcol which uses $O(k)$ space, and obtains $4$-approximation to the \fMM problem. 

The algorithm maintains $3$ sets $S, S_1, S_2$ using \tauGMM for all of them. In $S$, we include points in a group-agnostic way (similar to \textsc{Fair-Swap}) ignoring the fairness constraints. In $S_1$, we include points only of group $1$, and in $S_2$, we include points only of group $2$. By setting $\tau = \lsf/2$ we maintain the sets $S$, $S_1$ and $S_2$ such that all points are at least $\lsf/2$ distance apart in every one of them. 

Without loss of generality, suppose $\cX_1$ satisfies $|S \cap \cX_1| < k_1$. Our algorithm proceeds by identifying $k_1-|S \cap \cX_1|$ additional points from $S_1$ denoted by $Z_1$ by running \tauGMM with $\tau = \lsf/4$. This ensures that the final set of points from group $1$, i.e., $(S \cap \cX_1) \cup Z_1$ are $\lsf/4$ apart. By discarding the nearest neighbors of newly added points (i.e., $Z_1$), in $S \cap \cX_2$, we argue that our algorithm obtains a $4$-approximation. We obtain the following guarantees:

\begin{theorem}\label{thm:stream2group_informal}
There is a one-pass streaming algorithm that returns a $4(1+\epsilon)$-approximation for \fMM problem using $O(k\epsilon^{-1}\log n)$ space.
\end{theorem}

\subsection{Composable Coresets}\label{sec:composable_coresets} 

In this section, we design \emph{composable} coresets for \fMM. We assume the points $\cX$ are partitioned into $L$ disjoint sets. We discuss an algorithm for constructing $(1+\epsilon)$-composable coresets for Euclidean metrics, and discuss extensions. Missing details are presented in Appendix~\ref{app:composable_coresets}.

\smallskip

\subsubsection{Constructing $(1+\epsilon)$-composable coresets} We assume the universe of points $\cX$ is partitioned into
a collection of $L$ disjoint sets $\cY_1, \cY_2, \cdots, \cY_{L}$. As in Section~\ref{subsec:euclCoresets}, we define an $\foo > 0$ value such that $(1-\foo) = 1/(1+\appr)$.

We generalize the approach for constructing the coreset $\cT$ as follows: let
$\cY_j^i$ denote the points of group $i$ present in $\cY_j$ for $i \in [m]$
and $j \in [L]$. Then on each partition $j$ and group $i$, we run \gonzalez to
retrieve a diverse set $T_i^j$ with $O((4/\foo)^\doubldim k)$, or equivalently $O((8/\appr)^\doubldim k)$ points since $\foo \geq \appr/2$. The coreset $\cT$ is defined as:

\begin{enumerate}
	\item For $j \in [L]$, construct $T_j$: $T_j \leftarrow \bigcup_{i=1}^{m} T_{j}^{i}, \ \text{where} \ T_{j}^{i} \leftarrow \gonzalez(\cY_j^{i}, (4/\foo)^\doubldim k)$
	\item $\cT \leftarrow \bigcup\limits_{j=1}^{L} T_j$
\end{enumerate}

We obtain the following theorem:
\begin{theorem}\label{thm:comp_coresets_doubling}
$\cT$ is a $(1+\epsilon)$-composable coreset for fair Max-Min diversification
of size $O((8/\epsilon)^\doubldim km L)$ in metrics of doubling dimension
$\doubldim$ that can be obtained in $O((8/\epsilon)^\doubldim kmnL)$ time.
\end{theorem}

For general metrics, using a similar approach, we obtain a $5$-composable coreset by extending a recent construction of $5$-coreset for the sequential setting~\cite{moumoulidou21}. We also discuss two-pass distributed algorithms for constructing $\alpha$-composable coresets for Euclidean ($\alpha = 1+\epsilon$) and general metrics ($\alpha =5$).
\section{Conclusion}\label{sec:conclusions}

In this paper, we presented new approximation algorithms that substantially improve upon currently known results for the  \fMM problem both in
general and Euclidean metric spaces. 
There are several interesting directions for future work, including obtaining a $2$-approximation for the problem in general metrics or improving the hardness result.

Another direction is to generalize the fairness constraints to arbitrary matroid constraints (the fairness constraints considered in this paper can be expressed via the special case of a partition matroid). While there are results known for related diversity maximization problems under matroid constraints~\cite{abbassi2013, Bhaskara16, Borodin2012}, to the best of our knowledge, there are currently no results for Max-Min diversification.

\subsection*{Acknowledgements} This work was supported by the NSF under grants CCF-1934846, CCF-1908849, CCF-1637536, CCF-1763423, IIS-1943971, and an Adobe Research Grant.

\bibliography{literature_arxiv} 

\begin{thebibliography}{10}

\bibitem{abbar2013diverse}
Sofiane Abbar, Sihem Amer-Yahia, Piotr Indyk, Sepideh Mahabadi, and Kasturi~R
  Varadarajan.
\newblock Diverse near neighbor problem.
\newblock In {\em Proceedings of the twenty-ninth annual symposium on
  Computational geometry}, pages 207--214, 2013.

\bibitem{abbassi2013}
Zeinab Abbassi, Vahab~S. Mirrokni, and Mayur Thakur.
\newblock Diversity maximization under matroid constraints.
\newblock In {\em KDD '13}, pages 32--40, 2013.

\bibitem{Agarwal2007GeometricAV}
P.~Agarwal, Sariel Har-Peled, and Kasturi~R. Varadarajan.
\newblock Geometric approximation via coresets.
\newblock 2007.

\bibitem{10.1145/2213556.2213562}
Pankaj~K. Agarwal, Graham Cormode, Zengfeng Huang, Jeff Phillips, Zhewei Wei,
  and Ke~Yi.
\newblock Mergeable summaries.
\newblock In {\em PODS '12}, page 23–34, 2012.

\bibitem{Agarwal2020}
Pankaj~K. Agarwal, Stavros Sintos, and Alex Steiger.
\newblock Efficient indexes for diverse top-k range queries.
\newblock In {\em PODS '20}, page 213–227, 2020.

\bibitem{Aghamolaei2015}
Sepideh Aghamolaei, Majid Farhadi, and Hamid Zarrabi-Zadeh.
\newblock Diversity maximization via composable coresets.
\newblock In {\em CCCG}, 2015.

\bibitem{AngelK2011}
Albert Angel and Nick Koudas.
\newblock Efficient diversity-aware search.
\newblock In {\em SIGMOD ’11}, page 781–792, 2011.

\bibitem{pcp}
Sanjeev Arora, Carsten Lund, Rajeev Motwani, Madhu Sudan, and Mario Szegedy.
\newblock Proof verification and the hardness of approximation problems.
\newblock {\em J. ACM}, 45(3):501–555, May 1998.

\bibitem{assouad1983plongements}
Patrice Assouad.
\newblock Plongements lipschitziens dans $\mathbb{R}^n$.
\newblock {\em Bulletin de la Soci{\'e}t{\'e} Math{\'e}matique de France},
  111:429--448, 1983.

\bibitem{bera2019fair}
Suman Bera, Deeparnab Chakrabarty, Nicolas Flores, and Maryam Negahbani.
\newblock Fair algorithms for clustering.
\newblock {\em Advances in Neural Information Processing Systems},
  32:4954--4965, 2019.

\bibitem{Bhaskara16}
Aditya Bhaskara, Mehrdad Ghadiri, Vahab Mirrokni, and Ola Svensson.
\newblock Linear relaxations for finding diverse elements in metric spaces.
\newblock In {\em NIPS’16}, page 4105–4113, 2016.

\bibitem{Borassi19}
Michele Borassi, Alessandro Epasto, Silvio Lattanzi, Sergei Vassilvitskii, and
  Morteza Zadimoghaddam.
\newblock Better sliding window algorithms to maximize subadditive and
  diversity objectives.
\newblock In {\em PODS '19}, page 254–268, 2019.

\bibitem{Borodin2017}
Allan Borodin, Aadhar Jain, Hyun~Chul Lee, and Yuli Ye.
\newblock Max-sum diversification, monotone submodular functions, and dynamic
  updates.
\newblock {\em ACM Trans. Algorithms}, 2017.

\bibitem{Borodin2012}
Allan Borodin, Hyun~Chul Lee, and Yuli Ye.
\newblock Max-sum diversification, monotone submodular functions and dynamic
  updates.
\newblock In {\em PODS '12}, pages 155--166, 2012.

\bibitem{mmr}
Jaime Carbonell and Jade Goldstein.
\newblock The use of mmr, diversity-based reranking for reordering documents
  and producing summaries.
\newblock In {\em SIGIR ’98}, page 335–336, 1998.

\bibitem{Ceccarello2018}
Matteo Ceccarello, Andrea Pietracaprina, and Geppino Pucci.
\newblock Fast coreset-based diversity maximization under matroid constraints.
\newblock In {\em WSDM '18}, pages 81--89, 2018.

\bibitem{Ceccarello2021}
Matteo Ceccarello, Andrea Pietracaprina, and Geppino Pucci.
\newblock A general coreset-based approach to diversity maximization under
  matroid constraints.
\newblock {\em ACM Trans. Knowl. Discov. Data}, 2020.

\bibitem{Ceccarello17}
Matteo Ceccarello, Andrea Pietracaprina, Geppino Pucci, and Eli Upfal.
\newblock Mapreduce and streaming algorithms for diversity maximization in
  metric spaces of bounded doubling dimension.
\newblock {\em Proc. VLDB Endow.}, page 469–480, 2017.

\bibitem{pmlr-v80-celis18a}
Elisa Celis, Vijay Keswani, Damian Straszak, Amit Deshpande, Tarun Kathuria,
  and Nisheeth Vishnoi.
\newblock Fair and diverse {DPP}-based data summarization.
\newblock In {\em ICML '2018}, pages 716--725, 2018.

\bibitem{Celis2017RankingWF}
L.~Elisa Celis, Damian Straszak, and Nisheeth~K. Vishnoi.
\newblock Ranking with fairness constraints.
\newblock In {\em ICALP}, 2017.

\bibitem{Cevallos17}
Alfonso Cevallos, Friedrich Eisenbrand, and Rico Zenklusen.
\newblock Local search for max-sum diversification.
\newblock In {\em SODA ’17}, page 130–142, 2017.

\bibitem{chandra2001}
Barun Chandra and Magn\'{u}s~M Halld\'{o}rsson.
\newblock Approximation algorithms for dispersion problems.
\newblock {\em J. Algorithms}, pages 438--465, 2001.

\bibitem{chierichetti2017fair}
Flavio Chierichetti, Ravi Kumar, Silvio Lattanzi, and Sergei Vassilvitskii.
\newblock Fair clustering through fairlets.
\newblock In {\em NIPS}, 2017.

\bibitem{chiplunkar2020solve}
Ashish Chiplunkar, Sagar Kale, and Sivaramakrishnan~Natarajan Ramamoorthy.
\newblock How to solve fair k-center in massive data models.
\newblock In {\em {ICML} 2020}, pages 1877--1886, 2020.

\bibitem{cormen2009introduction}
Thomas~H Cormen, Charles~E Leiserson, Ronald~L Rivest, and Clifford Stein.
\newblock {\em Introduction to algorithms}.
\newblock MIT press, 2009.

\bibitem{CormodeMZ07}
Graham Cormode, S.~Muthukrishnan, and Wei Zhuang.
\newblock Conquering the divide: Continuous clustering of distributed data
  streams.
\newblock In {\em {ICDE} 2007}, pages 1036--1045, 2007.

\bibitem{deng2014complexity}
Ting Deng and Wenfei Fan.
\newblock On the complexity of query result diversification.
\newblock {\em ACM Transactions on Database Systems (TODS)}, 39(2):1--46, 2014.

\bibitem{Drosou2014}
M.~{Drosou} and E.~{Pitoura}.
\newblock Diverse set selection over dynamic data.
\newblock {\em IEEE Transactions on Knowledge and Data Engineering},
  26(5):1102--1116, 2014.

\bibitem{article}
Marina Drosou, H.V. Jagadish, Evaggelia Pitoura, and Julia Stoyanovich.
\newblock Diversity in big data: A review.
\newblock {\em Big Data}, 5:73--84, 2017.

\bibitem{drosou2010}
Marina Drosou and Evaggelia Pitoura.
\newblock Search result diversification.
\newblock {\em SIGMOD Rec.}, (1):41--47, 2010.

\bibitem{drosoudisc}
Marina Drosou and Evaggelia Pitoura.
\newblock Disc diversity: Result diversification based on dissimilarity and
  coverage.
\newblock {\em Proc. VLDB Endow.}, 6(1):13–24, November 2012.

\bibitem{Dwork:2012:FTA:2090236.2090255}
Cynthia Dwork, Moritz Hardt, Toniann Pitassi, Omer Reingold, and Richard Zemel.
\newblock Fairness through awareness.
\newblock In {\em ITCS '12}, pages 214--226, 2012.

\bibitem{halabi2021}
Marwa El~Halabi, Slobodan Mitrovi\'{c}, Ashkan Norouzi-Fard, Jakab Tardos, and
  Jakub~M Tarnawski.
\newblock Fairness in streaming submodular maximization: Algorithms and
  hardness.
\newblock In {\em NeurIPS 2020}, volume~33, pages 13609--13622, 2020.

\bibitem{ERKUT199048}
Erhan Erkut.
\newblock The discrete p-dispersion problem.
\newblock {\em European Journal of Operational Research}, 46(1):48 -- 60, 1990.

\bibitem{GalhotraBM2017}
Sainyam Galhotra, Yuriy Brun, and Alexandra Meliou.
\newblock Fairness testing: Testing software for discrimination.
\newblock In {\em ESEC/FSE '17}, pages 498--510, 2017.

\bibitem{gollapudi2009}
Sreenivas Gollapudi and Aneesh Sharma.
\newblock An axiomatic approach for result diversification.
\newblock In {\em WWW ’09}, page 381–390, 2009.

\bibitem{Gonzalez1985ClusteringTM}
Teofilo~F. Gonzalez.
\newblock Clustering to minimize the maximum intercluster distance.
\newblock {\em Theor. Comput. Sci.}, 38:293--306, 1985.

\bibitem{Guha2009}
Sudipto Guha.
\newblock Tight results for clustering and summarizing data streams.
\newblock In {\em ICDT '09}, page 268–275, 2009.

\bibitem{gupta2003bounded}
Anupam Gupta, Robert Krauthgamer, and James~R Lee.
\newblock Bounded geometries, fractals, and low-distortion embeddings.
\newblock In {\em 44th Annual IEEE Symposium on Foundations of Computer
  Science, 2003. Proceedings.}, pages 534--543. IEEE, 2003.

\bibitem{hassin1997}
Refael Hassin, Shlomi Rubinstein, and Arie Tamir.
\newblock Approximation algorithms for maximum dispersion.
\newblock {\em Oper. Res. Lett.}, 21(3):133--137, October 1997.

\bibitem{Indyk14}
Piotr Indyk, Sepideh Mahabadi, Mohammad Mahdian, and Vahab~S. Mirrokni.
\newblock Composable core-sets for diversity and coverage maximization.
\newblock In {\em PODS ’14}, page 100–108, 2014.

\bibitem{Jones2020FairKV}
Matthew Jones, Huy Nguyen, and Thy Nguyen.
\newblock Fair k-centers via maximum matching.
\newblock In {\em {ICML} 2020}, pages 4940--4949, 2020.

\bibitem{pmlr-v97-kleindessner19a}
Matth{\"a}us Kleindessner, Pranjal Awasthi, and Jamie Morgenstern.
\newblock Fair k-center clustering for data summarization.
\newblock In {\em ICML '19}, volume~97, pages 3448--3457, 09--15 Jun 2019.

\bibitem{Kuby1987}
Michael~J. Kuby.
\newblock Programming models for facility dispersion: The p-dispersion and
  maxisum dispersion problems.
\newblock {\em Geographical Analysis}, 19(4):315--329, 1987.

\bibitem{moumoulidou21}
Zafeiria Moumoulidou, Andrew McGregor, and Alexandra Meliou.
\newblock {Diverse Data Selection under Fairness Constraints}.
\newblock In {\em ICDT 2021}, pages 13:1--13:25, 2021.

\bibitem{QinYC2012}
Lu~Qin, Jeffrey~Xu Yu, and Lijun Chang.
\newblock Diversifying top-k results.
\newblock {\em Proc. VLDB Endow.}, 5(11):1124–1135, July 2012.

\bibitem{ravi1994}
S.~S. Ravi, D.~J. Rosenkrantz, and G.~K. Tayi.
\newblock Heuristic and special case algorithms for dispersion problems.
\newblock {\em Oper. Res.}, 42(2):299--310, April 1994.

\bibitem{schrijver2003}
Alexander Schrijver.
\newblock {\em Combinatorial optimization: polyhedra and efficiency},
  volume~24.
\newblock Springer Science \& Business Media, 2003.

\bibitem{Stoyanovich2018OnlineSS}
Julia Stoyanovich, Ke~Yang, and H.~V. Jagadish.
\newblock Online set selection with fairness and diversity constraints.
\newblock In {\em EDBT}, 2018.

\bibitem{Tamir}
Arie Tamir.
\newblock Obnoxious facility location on graphs.
\newblock {\em SIAM J. Discrete Math.}, 4:550--567, 11 1991.

\bibitem{viera11}
Marcos~R. Vieira, Humberto~L. Razente, Maria C.~N. Barioni, Marios
  Hadjieleftheriou, Divesh Srivastava, Caetano Traina, and Vassilis~J. Tsotras.
\newblock On query result diversification.
\newblock In {\em ICDE 2011}, pages 1163--1174, 2011.

\bibitem{wang2021}
Yanhao Wang, Francesco Fabbri, and Michael Mathioudakis.
\newblock Fair and representative subset selection from data streams.
\newblock In {\em WWW 2021}, page 1340–1350, 2021.

\bibitem{Wang2018}
Yue Wang, Alexandra Meliou, and Gerome Miklau.
\newblock Rc-index: Diversifying answers to range queries.
\newblock {\em Proc. VLDB Endow.}, 11(7):773–786, 2018.

\bibitem{Yang:2019:BRD:3367722.3367886}
Ke~Yang, Vasilis Gkatzelis, and Julia Stoyanovich.
\newblock Balanced ranking with diversity constraints.
\newblock In {\em IJCAI'19}, pages 6035--6042, 2019.

\bibitem{YangS2017}
Ke~Yang and Julia Stoyanovich.
\newblock Measuring fairness in ranked outputs.
\newblock In {\em SSDBM ’17}, 2017.

\bibitem{Zehlike:2017:FFT:3132847.3132938}
Meike Zehlike, Francesco Bonchi, Carlos Castillo, Sara Hajian, Mohamed Megahed,
  and Ricardo Baeza-Yates.
\newblock Fa*ir: A fair top-k ranking algorithm.
\newblock In {\em CIKM '17}, pages 1569--1578, 2017.

\end{thebibliography}

\newpage
\section*{Appendix}
\appendix
\section{$(m+1)$-Approximation with Perfect Fairness}\label{app:lineartime} 

We give the pseudocode for \genswap (Algorithm~\ref{alg:genswap}) described in Section~\ref{subsec:lineartime} and analyze its approximation guarantees.  Let $\mathcal{C}$ denote the set of clusters obtained when the input $\gamma$ (a
guess for the optimal diversity value) to Algorithm~\ref{alg:genswap} satisfies,
$\lsf/(1+\epsilon) < \gamma \leq \lsf$ for some $\epsilon > 0$. Let $\cS^* = \{
y_1, y_2, \cdots y_k \}$ denote the set of points selected by an optimal solution. 

\changes{Recall that by construction of a cluster $D \in \mathcal{C}$, any point $p \in \cX$ from some group $i$ is removed from further consideration only if: (1)~it is less than $\frac{\gamma}{m+1}$ distance from some point in $D$, \emph{and} (2)~there is already a point $p' \in D$ of group $i$. Using Claim~\ref{cl:atmost_one}, we show that in every cluster $D \in \mathcal{C}$ there is at most one optimal point $y_i \in \cS^{*}$, which is either included in $D$, or if $y_i$ was removed, then, there is a point $y' \in D$ of the same group as $y_i$.}

\changes{\smallskip
\noindent \textbf{Notation}. Define $f(y_i) \in D$ to be a point of the same group as $y_i$, and note that $f(y_i)$ could be $y_i$ itself. Let $f(\cS^*)= \{f(y_1), f(y_2), \cdots, f(y_k)\}$.

To be more specific, assume an optimal point $y_i \in \cS^*$ is less than $\frac{\gamma}{m+1}$ distance from some point in $D$. Let $f(y_i)$ represent a point in $D$ of the same group as $y_i$, and note that $f(y_i)$ could be $y_i$ itself if $y_i$ was not removed while constructing $D$; otherwise $f(y_i) = y'$. Then, we show that there cannot be a pair of $f(y_i), f(y_j)$ points in the same cluster $D \in \mathcal{C}$ for any pair of optimal points $y_i, y_j \in \cS^{*}$. Consequently, we can argue that there are at least $k_i$ points from each group $i$ in the formed clusters $\mathcal{C}$, some of which could be the optimal points themselves, and are all in \emph{different} clusters. This will let us argue it is possible to find a flow of size $k$ using the reduction to max-flow.}

\begin{algorithm}[!ht] 
    \caption{\genswap}\label{alg:genswap} 
    {\footnotesize
        \begin{algorithmic}[1]
            \Statex
            \begin{description} 
            \item[\rlap{Input:}\phantom{Output:}] $\cX = \bigcup\limits_{i=1}^{m} \cX_i$: Universe of available elements. 
                \item[\phantom{Output:}] $k_1, \ldots, k_m \in \mathbb{Z}^{+}$.
                \item[\phantom{Output:}] $\gamma\in \mathbb{R}^+$: A guess of the optimum fair diversity.
                \item[Output:] $k_i$ points in $\cX_i$ for $i\in [m]$. 
                \end{description} 
       \State $\mathcal{R} \leftarrow \cX$ denote the set of remaining elements.
       \State $\mathcal{C} \leftarrow \emptyset$ denote a collection of subsets of points (called clusters).
       \While{$|\mathcal{R}| > 0$ \textbf{ (and) } $|\mathcal{C}| \leq k m$}
            \State $D \leftarrow \emptyset$ denote the current cluster, and $D_{\text{col}} \leftarrow \emptyset$ denote the groups of points in cluster $D$.
            \While{an element $p \in \mathcal{R} \cap \cX_i$ for some $i \in \{ 1, 2, \cdots , m\}\setminus D_{\text{col}}$. exists}
            \If{$|D| = 0$ (or) $d(p, x) < \frac{\gamma}{m+1}$ for some $x \in D$}
                \State $D \leftarrow D \cup \{p\}$ and $D_{\text{col}} \leftarrow D_{\text{col}} \cup \{ i\}$.
            \EndIf
            \EndWhile
            \State $\mathcal{R} \leftarrow \mathcal{R} \setminus \bigcup_{p \in D} \rball{p}{\frac{\gamma}{m+1}}.$
    \State $ \mathcal{C} \leftarrow \mathcal{C} \cup \{ D \}$.
    \State $\cR \leftarrow \cR \setminus \cX_i \ \forall i \in [m]$ if $| \{ D \mid D \in \mathcal{C} \text{ and } D \cap \cX_i \neq \emptyset \} | \geq k$.
       \EndWhile
         \LineCommentx{Construct flow graph :}
        \State Let $\mathcal{C} = \{ D_1, D_2, \cdots D_t\}$.
        \State 
         Construct directed graph $G=(V,E)$ where
        \begin{eqnarray*}
        V&=& \{a,u_1, \ldots, u_m, v_1, \ldots, v_t, b\}  \\
        E&=& \{(a,u_i) \mbox{ with capacity $k_i$}: i \in [m]\} \\
        & & \cup~ \{(v_j,b) \mbox{ with capacity $1$}: j \in [t]\} \\
        & & \cup~  \{(u_i,v_j) \mbox{ with capacity $1$}:  |\cX_i\cap D_j|\geq 1\} 
        \end{eqnarray*} 
        
        \State Set $\cS \leftarrow \emptyset$. Compute maximum $a$-$b$ flow in $G$ using Ford-Fulkerson algorithm~\cite{cormen2009introduction}. 
        \If{flow size $<k=\sum_i k_i$}
        \Return $\emptyset$ \Comment{Abort}
        \Else \Comment{max flow is $k$}
        \State $\forall(u_i,v_j)$ with flow equal to $1$, add the point in $D_j$ with group $i$ to $\mathcal{S}$.
        \EndIf
        \State \Return $\mathcal{S}$.
        \end{algorithmic}} 
\end{algorithm} 

\changes{
\begin{claim}\label{cl:atmost_one} Let $y_i \in \cS^{*}$ denote an optimal point that is less than $\frac{\gamma}{m+1}$ distance from some point in some cluster $D \in \mathcal{C}$. We have: $|\bigcup_{p \in D} \mathbf{B}(p, \frac{\gamma}{m+1}) \cap f(\cS^*)| \leq 1$ for every cluster $D \in \mathcal{C}$.
\end{claim}}

\changes{
\begin{proof}
For the sake of contradiction, assume there is pair of points $$f(y_1), f(y_2) \in \bigcup_{p \in D} \mathbf{B}(p, \frac{\gamma}{m+1}) \cap f(\cS^*)$$ 

for a pair of optimal points $y_1, y_2 \in \cS^{*}$. From the definition of
optimality, we have: $d(y_1, y_2) \geq \lsf \geq \gamma.$ 

Let $p^{y_1} \in D$ and $p^{y_2} \in D$ be the points from which $y_1$ (and $y_2$ respectively), is $<\frac{\gamma}{m+1}$ distance apart. Note that it is not necessarily the case that $p^{y_1} = f(y_1)$ (or $ p^{y_2} = f(y_2)$). As $D$ contains at most $m$ points, note that the largest chain in $D$ has length of $m-1$ and thus by triangle inequality and construction of $D$, $d(p^{y_1}, p^{y_2}) \leq \frac{(m-1) \gamma}{m+1}.$ Once again using triangle inequality, we have:
\begin{align*}
	d(y_1, y_2) &\leq d(y_1, p^{y_1}) + d(y_2, p^{y_2}) + d(p^{y_1}, p^{y_2})\\
	&< \frac{\gamma}{m+1} + \frac{\gamma}{m+1} + \frac{(m-1)\gamma}{m+1} < 	\gamma. 
\end{align*}

This is a contradiction. Consequently, it follows that there is at most one optimal point $y_i \in \cS^{*}$ in every cluster $D \in \mathcal{C}$, and it is either included in $D$, in which case $f(y_i) = y_i$, or there is another point $y'$ from the same group as $y_i$ present in $D$ (i.e., $f(y_i)=y'$). 
\end{proof}
}

Now, we argue that the set of all points in $\mathcal{C}$ contains a subset of
points that obtains an ${m+1}$-approximation factor.
\begin{claim}\label{cl:subset_flow}
There exists a subset $S \subseteq \bigcup_{D \in \mathcal{C}} D$ such that
\begin{itemize}
	 \item[(i)] $|S \cap D| \leq 1$ for every $D \in \mathcal{C}$ and $|S \cap \cX_j| = k_j$ for every $j \in [m]$.
	 \item[(ii)] For every pair of points $s_1, s_2 \in S$, $\ d(s_1, s_2) \geq \frac{\gamma}{m+1}$.
\end{itemize}
\end{claim}

\begin{proof}	
\changes{

\noindent 
\textbf{(i)} We show that in the formed clusters $\mathcal{C}$, it possible to select $k_j$ points from each group $j$ by selecting at most one point from each cluster $D \in \mathcal{C}$.

Following the analysis in~\cite{moumoulidou21}, call a group $j$ \emph{non-critical} if there are $k$ points in the formulated set of clusters $\mathcal{C}$. Otherwise, call group $j$ \emph{critical}. Note that if a group $j$ is \emph{critical}, \genswap processed all the points in $\cX_j$. For non-critical groups, recall that the remaining points are discarded once we find $k$ points.

For a critical group $j$, from Claim~\ref{cl:atmost_one} and the fact that all points were processed while constructing the clusters, we know that every cluster $D \in \mathcal{C}$ satisfies $|\bigcup_{p \in D} \mathbf{B}(p, \frac{\gamma}{m+1}) \cap f(\cS_{j}^*)| \leq 1$. Therefore, it follows that for a critical group $j$ there are $k_j$ distinct clusters among the formed clusters that contain the points in $f(\cS{_j}^{*})$. We add the points in $f(\cS{_j}^{*})$ to $S$.

Note that for all \emph{non-critical} groups, the selected points are all in different clusters by construction. Thus, regardless of how pick the points for the different groups across the clusters, by the time we process some non-critical group $j$ there will be at least $k-\sum_{i:\textup{processed\_groups}} k_i \geq k_j$ clusters, each of which contain one point from group $j$. These $k_j$ points can be added to $S$ by picking only one point from each cluster.

Thus, it follows that we can pick $k_j$ points for each group $j \in [m]$ by picking at most one point from each cluster $D \in \mathcal{C}$.

\noindent \textbf{(ii)} Consider two points $s_1, s_2 \in S$. As $|S \cap D|
\leq 1$ for every $D \in \mathcal{C}$, we have that $s_1$ and $s_2$ belong to
different clusters, say $D^{s_1}$ and $D^{s_2}$. Without loss of generality suppose the cluster $D^{s_1}$ is formed before $D^{s_2}$. From Algorithm~\ref{alg:genswap}, we know that all the points in $\mathbf{B}(s_1, \frac{\gamma}{m+1})$ are removed, and $s_2$ is chosen from the remaining points. Therefore, $d(s_1, s_2) \geq \frac{\gamma}{m+1}$.
}
\end{proof}

\begin{claim}\label{cl:validflow}
There is a valid flow of size $k$ in the directed graph $G$ as described in \genswap (Algorithm~\ref{alg:genswap}).
\end{claim}

\begin{proof}
Let $\mathcal{C}$ be the clustering obtained by \genswap. Consider the sum of the capacities of all edges incident on the source node $a$ in the
graph $G$. It is given by $\sum_{i \in [m]} k_i = k$. Therefore, the maximum flow in $G$ is at most $k$. We argue that we can construct a valid flow of size $k$ from $\mathcal{C}$.

Let $S\subseteq \bigcup_{D \in \mathcal{C}} D$ be a subset satisfying  Claim~\ref{cl:subset_flow}. For every point $e \in S$ such that $e \in D_j
\cap \cX_i$ for some $D_j \in \mathcal{C}, i \in [m]$, we add a flow of
value $1$ on the edges $(u_i, v_j)$, $(v_j, b)$ and $(a, u_i)$. As $|S \cap
\cX_i| = k_i$ (Claim~\ref{cl:subset_flow} (i)), we have a flow of
value $k_i$ on every edge of the form $(a, u_i)$ for $i \in [m]$. As $|S \cap
D_j| \leq 1$ (Claim~\ref{cl:subset_flow} (i)), the capacity constraints are
always satisfied for the edges incident on node $v_j$ for every $D_j \in
\mathcal{C}$. As $|S \cap \cX_i| = k_i$ (Claim~\ref{cl:subset_flow}
(i)), the capacity constraints are satisfied for the edges incident on node
$u_i$ for all $i \in [m]$. Therefore, there is a valid flow achieving the
maximum flow $\sum_{i \in [m]} k_i = k$ in the graph $G$.
\end{proof}

\begin{claim}\label{cl:approx}
Given a flow $f$ of size $k$ in $G$, the set $S$ obtained by
\genswap (Algorithm~\ref{alg:genswap}) satisfies the fairness constraints, and for every two points $s_1, s_2 \in S$, we have $d(s_1, s_2) \geq \frac{\gamma}{m+1}$.
\end{claim}

\begin{proof}
As the sum of the capacities of all edges incident on node $a$ in the graph $G$
is $\sum_{i \in [m]} k_i = k$, we can conclude that an edge incident on node
$u_i$ has a flow of $k_i$ in $f$. From conservation of flow, at every node
$u_i$, there are $k_i$ edges of the form $(u_i, v_j)$ with flow $1$. We include
a node from $D_j$ of group $i$ in $S$ if the flow on the edge $(u_i, v_j)$ is
$1$. Therefore, we have included $k_i$ nodes of group $i$ in $S$ for every $i
\in [m]$ satisfying the fairness constraints.

Consider $s_1, s_2 \in S$ belonging to clusters $D^{s_1}$ and $D^{s_2}$
respectively. From the construction, every edge $(v_j, b)$ of capacity $1$
ensures that only one point from a cluster in $\mathcal{C}$ is picked.
Therefore, we have $D^{s_1}$ and $D^{s_2}$ are distinct clusters. Without loss
of generality let the cluster $D^{s_1}$ be formed before $D^{s_2}$. From
Algorithm~\ref{alg:genswap}, we know that all the points in $\mathbf{B}(s_1,
\frac{\gamma}{m+1})$ are removed, and $s_2$ is chosen from the remaining
points. Therefore, $d(s_1, s_2) \geq \frac{\gamma}{m+1}$.
\end{proof}

\begin{theorem}[Theorem~\ref{thm:lineartime} restated]
\genswap (Algorithm~\ref{alg:genswap}) returns a $(m+1) (1+\epsilon)$-approximation and achieves {perfect fairness} for the \fMM problem using a running time of $O(nkm^3 \epsilon^{-1}\log n)$.
\end{theorem}
\begin{proof}

Suppose the guess $\gamma \geq \lsf/(1+\epsilon)$ is obtained using binary search. Let $\mathcal{C}$ denote the clusters obtained by Algorithm~\ref{alg:genswap} using the
guess $\gamma$. From Claim~\ref{cl:validflow}, we have that there is a valid
flow of size $k$ in $G$. After solving the maximum flow problem on the graph
$G$, let $f$ be the flow. Using Claim~\ref{cl:approx}, we can obtain a set $S$
from $f$ that satisfies fairness constraints, and for every $s_1, s_2 \in S$,
we have: $d(s_1, s_2) \geq \frac{\gamma}{m+1} = \frac{\lsf}{(m+1)(1+\epsilon)}.$
Hence, $S$ obtains an approximation factor of $(m+1)(1+\epsilon)$.

\smallskip

\noindent \textit{Running Time}. 
Each cluster $D \in \mathcal{C}$ contains at most $m$ points. For including a
point $p \in D$, we iterate over all remaining points $\mathcal{R}$ that are
within a distance of $\frac{\gamma}{m+1}$, and of a group different from the
points already included in $D$. This step requires a running time of
$O(|\mathcal{R}|\cdot m \cdot m) = O(nm^2)$. We remove all points of group $i \in [m]$ from $\cR$, if the total number of clusters containing a point from the group is at least $k$. This ensures that there are at most $km$ clusters that are possible.
Therefore, the total running time for constructing $\mathcal{C}$ is $O(nm^2 \cdot km) = O(nkm^3)$. \\

\noindent We can observe that the number of edges in $G$ are $O(|\mathcal{C}|m)
= O(km^2)$. Using Ford-Fulkerson, we can obtain the maximum flow of value $k$. This step requires a running time of $O(k^2m^2)$. For identifying the correct guess $\gamma$, we use binary searching which results in a multiplicative factor of $O(\log n/\epsilon)$ to the total running time  (see Section~\ref{subsec:randomLP} for a discussion). Combining all the running times, gives us the lemma.
\end{proof}

\subsection{Tight Example for $3m-1$ approximation algorithm in~\cite{moumoulidou21}}\label{app:tightexample}

\changes{Below we show a tight example for \fMMFlow in~\cite{moumoulidou21} and show how \genswap yields a better approximation. 
\smallskip

\noindent 
\textbf{A tight example for \fMMFlow.} Suppose $k=3$ and we have to select one white and two black points. Here, edges represent the distance across two points, e.g., $d(p_1, p_2)= 1/5$. Note that the optimal solution in this example is the set of points $\{p_1, p_3, p_4\}$ with diversity score equal to $1$. 

\centering
\includegraphics[width=0.5\textwidth]{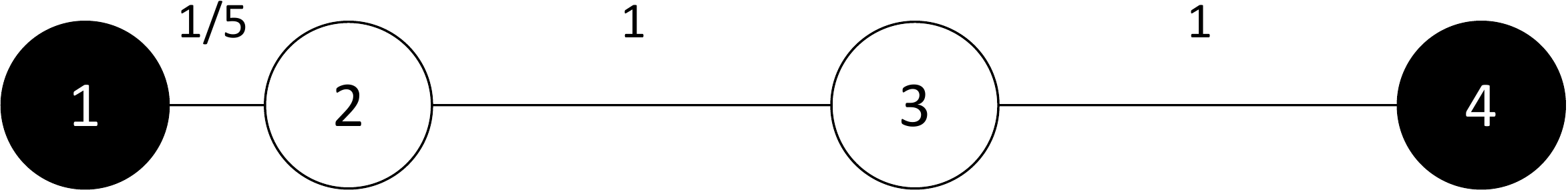}

The \fMMFlow algorithm in~\cite{moumoulidou21}, for a guess $\gamma=1$, for the black group selects both points since they are at least $d_1=\frac{m\gamma}{3m-1}=2/5$ far apart from each other. Similarly for the white group. Now because there is no pair of points with distance strictly less than $d_2= \frac{\gamma}{3m-1}=1/5$, \fMMFlow constructs four connected components (each with a point). As a result, the points $\{p_1, p_2, p_4\}$ will be selected by the max-flow algorithm and we obtain a set with diversity score equal to $1/5$. Note that for this example, \genswap returns the set $\{p_1, p_3, p_4\}$ as $p_1$ and $p_2$ are less than $1/3$ distance apart. These two points will be in the same cluster and at most one of them can be picked; thus, we guarantee an approximation ratio of $3$.}

\section{Hardness of Approximation}\label{app:genmetrics}

\begin{theorem}[Theorem~\ref{thm:hardness} restated]
Unless $P = NP$, there is no {polynomial time} algorithm for the \fMM problem that obtains an $\alpha$-approximation factor for diversity score, and a
$\beta$-approximation for the fairness constraints, i.e., the set $\cS$
returned by the algorithm satisfies $|\cS \cap \cX_i| \geq {k_i}/{\beta} \
\forall i \in [m]$, for some constants $\alpha < 2$ and $\beta > 0$.
\end{theorem}

\begin{proof}
We present a reduction from $\textsc{Gap-Clique}_\rho$, where $\rho = \beta$.
For every vertex of the graph $G$, we create a new point, and set of points is
denoted by $\cX$. For every edge $(u, v)$ in $G$, we set $d(u, v) := 2$. For
all other pairs of vertices, we set the distances as $1$. Every vertex is
assigned the same color, and the corresponding fairness constraint is $|\cS
\cap \cX| \geq k$, where $\cS$ is the set of points whose diversity we are
trying to maximize in \fMM.

Suppose there is a polynomial time algorithm that returns a set $\cS$, obtains
an $\alpha$-approximation for the diversity score, and a $\beta$-approximation
for the fairness constraints. We first consider the `\textsc{Yes}' instance in \textsc{Gap-Clique}$_\beta$, i.e., we assume there is a clique of size $k$ in $G$. This implies $\lsf = 2$. As $\alpha < 2$, we have that the set $\cS$
returned has a diversity score $\geq \lsf/\alpha > 1$. Therefore, $\cS$ is a
clique in $G$ as all other pairwise distances are $1$ (from construction). As $\cS$ is a $\beta$-approximation for the fairness constraint, we have that
$|\cS| \geq {k}/{\beta}$. Let us now consider the `\textsc{No}' instance, i.e.,
there is no clique of size $\geq k/\beta$ in $G$. Therefore, $|\cS \cap \cX| <
k/\beta$, as $|\cS \cap \cX|$ is upper bounded by the maximum clique size in
$G$. From the above arguments, we have that using our algorithm, we can distinguish the `\textsc{Yes}' and `\textsc{No}' instances of \textsc{Gap-Clique}$_\beta$, which is not possible unless $P = NP$~\cite{pcp}. Hence, the theorem.
\end{proof}

\section{Euclidean Metrics}\label{app:eucmetrics}

\subsection{Exact Computation in One Dimension}
\begin{theorem}[Theorem~\ref{thm:fairline} restated]\label{thm:fairline_app}
There is an algorithm that solves the \fMM problem exactly when the points can be embedded on a line and requires a running time of $O(n^4 \prod_{i=1}^m (k_i+1) )$. 
\end{theorem}

\begin{proof}
We use \fL to identify the exact solution. We observe that any optimal solution can be expressed as a subset of the first $j$ points for some $j \in [n]$. From the construction, if the guess $\gamma \leq \lsf$ there will always be at least $k_i$ points from group $i$ for all $i \in [m]$ that are all $\gamma$ far apart. Therefore, since the dynamic programming approach finds all the subsets with $k_i$ points per group $i$ for all $j \in [n]$, at least one of the $\tb[k_1, k_2, \cdots, k_m, j]$ entries will be equal to $1$ as required. As discussed previously, we can backtrack and construct the solution set. 

\emph{Running Time}. For a fixed guess $\gamma$, we need to compute $\prod_{i=1}^m (k_i+1)$ entries for every point, as every $k'_i$ for $i\in [m]$ takes at most $k_i + 1$ values. To compute an entry $\tb[\cdot, \cdot, \cdots, \cdot, j]$ using \fL~(Algorithm~\ref{alg:dp1}), we need to retrieve $O(n)$ distances to find point $j'$ that is at least $\gamma$ far apart from point $j$. Thus, the total running is equal to $O(n^2 \prod_{i=1}^m (k_i+1) )$ since there are $O(n \prod_{i=1}^m (k_i+1) )$
entries in $\tb$ and the computational cost to fill each entry is $O(n)$. As there are $O(n^2)$ distance values the guess $\gamma$ can take, the total running time is $O(n^4\prod_{i=1}^m (k_i+1))$.
\end{proof}

\subsection{$(1+\epsilon)$ Approx with $(1-\epsilon)$ Fairness}

\begin{theorem}[Theorem~\ref{thm:fairDP} restated]
If $\gamma = \lsf$, then, \fDP (Algorithm~\ref{alg:dp2})
returns a set $\cS$ that satisfies $\diver(\cS) \geq \lsf$ and $|\cS \cap \cX_i| \geq k_i \ \forall i \in [m]$ and has a running time of $O(\prod_{i=1}^{m}(k_i +1)^{2} 2^R t)$ where $R = \max\{|C_1|, |C_2|, \cdots, |C_t|\}$. 
\end{theorem}

\begin{proof}
As $\gamma = \lsf$, the optimal set of points satisfy the fairness constraints. From the construction in \fDP, we will return a set $\cS$ that has diversity $\lsf$, and achieves perfect fairness.

\emph{Running Time}. Consider a value $j \in [t]$. There are $\prod_{i=1}^{m}(k_i+1)$ entries in the table $\tb$ corresponding to this value of $j$. For every $k'_i \in \{0, 1, \cdots k_i\}$ and every subset $R \subseteq C_j$ where $|R \cap \cX_i| = p'_i \ \forall i \in [m]$, we check if there is a valid subset of points satisfying fairness constraints using the condition mentioned in \fDP. Since there at most $\prod_{i=1}^{m}(k_i+1)$ ways to enumerate the $p'_i$ values (because $p'_i \leq k'_i$), the total time to compute entries corresponding to this $j$ value is $O(\prod_{i=1}^{m}(k_i+1)^2 2^R)$. Therefore, to compute all the entries in $\tb$ we need $O(\prod_{i=1}^{m}(k_i+1)^2 2^R t)$ time.
\end{proof}

\begin{lemma}[Lemma~\ref{lem:packing} restated]
For every $j \in [t]$, $|C_j| \leq m \cdot \left( 8mD^{3/2}/\epsilon^2 \right)^\lambda$ where $\lambda$ is doubling dimension of $\mathbb{R}^D$.
\end{lemma}

\begin{proof} 
Consider all points in $C_j$ that belong to group $i$, i.e., $C_j \cap
\widehat{T}_i$. From the construction of $\widehat{T}_i \subseteq {\cT}$, we have that every pair of points of the same group is separated by a distance at least $\epsilon \gamma/4$. Therefore, each point can be represented by a ball of radius $\epsilon \gamma/8$, and we want to count the maximum number of non-overlapping balls that can be packed inside the cube $C_j$. Observe that the length of the diagonal of $C_j$ is $W\sqrt{D}$, and the cube lies entirely in the ball of radius $W\sqrt{D}/2$ with center at the middle of the diagonal. We call this \emph{cube ball}. As Euclidean metrics are  doubling metrics, we can cover the \emph{cube ball} with overlapping balls of radius $\epsilon\gamma/8$ and the number of the  balls required is $\left(\frac{W\sqrt{D}/2}{\epsilon\gamma/8}\right)^\lambda$, where $\lambda = O(D)$ is the doubling dimension of $\mathbb{R}^D$.

We can observe that the total volume occupied by the overlapping balls is at least the volume occupied by the non-overlapping balls corresponding to the points and having the same radius. Therefore, we can upper bound the number of points using the total number of non-overlapping balls used to cover the \emph{cube ball}. As there are $m$ groups, we have that the total number of the points in $C_j$ is: $|C_j| \leq m \cdot \left(\frac{W\sqrt{D}/2}{\epsilon\gamma/8} \right)^\lambda = m \cdot \left( 8mD^{3/2}/\epsilon^2 \right)^\lambda.$
\end{proof}

\changes{
\begin{theorem}[Theorem~\ref{thm:euclidean} restated]
If $\gamma \geq \lsf/(1+\epsilon)$, \fB Algorithm returns a set $\cS$ such that $\diver{(\cS)}\geq \lsf/(1+\epsilon)$ and $|\cS \cap \cX_i| \geq  k_i(1-\epsilon) \ \forall i \in [m]$ with probability at least $1-\delta$. The running time of the algorithm is:
\[ O(n \cdot (8/\epsilon)^\doubldim km +  \prod_{i=1}^{m}(k_i +1)^{2} 2^{m (8m D^{3/2} / \epsilon^2)^\lambda} (8/\epsilon)^\doubldim km \log|\cT| \log(1/\delta)).\]
\end{theorem}

\begin{proof}
The running time of \fB (Algorithm~\ref{algo:m4}) depends on: (1)~the running time of constructing the coreset $\cT$ which is $O(nkm (8/\epsilon)^\lambda)$, where $\lambda$ is the doubling dimension, and (2)~the running time of \fDP (Algorithm~\ref{alg:dp2}) on the clusters for every guess $\gamma$.

From Theorem~\ref{thm:fairDP}, we know that \fDP has a running time of $O(\prod_{i=1}^{m}(k_i +1)^{2} 2^R t)$, where $t$ is the number of clusters and $R$ is the maximum size across all $t$ clusters. We upper bound the number of clusters by the coreset size. So, $t=O((8/\epsilon)^\lambda km)$. From Lemma~\ref{lem:packing}, we have $R = O(m (8m D^{3/2}/ \epsilon^2)^\lambda)$. Combining all the above, the final running time is:
$$O((8/\epsilon)^\doubldim kmn + \log|\cT| \log(1/\delta) \prod_{i=1}^{m}(k_i +1)^{2} 2^{m (8m D^{3/2} / \epsilon^2)^\lambda} (8/\epsilon)^\doubldim km).$$
\end{proof}	
}

\section{Scalable Algorithms: Data Stream Algorithms}\label{app:streaming}
In this section, we present $1$-pass data stream algorithms that obtain good approximation guarantees for the \fMM problem. First, we present a variant called \tauGMM, of the~\gonzalez algorithm from Section~\ref{subsec:euclCoresets}. Instead of finding the point that is farthest away, we use a threshold $\tau$ to identify the next point to be included in the set of diverse points. The threshold $\tau$ is set to $\gamma/2$, for a guess $\gamma$ of the optimal diversity score $\lsf$. We argue that this threshold variant of \gonzalez algorithm (called \tauGMM) also obtains the same $2$-approximation guarantee. 

Using \tauGMM, we give several data stream algorithms for the \fMM problem. In Section~\ref{app:streamgen}, we give an algorithm that obtains a $30$-approximation using $O(km)$ space for general metrics. In Section~\ref{app:streameuc}, we give an algorithm that obtains a $(1+\epsilon)$-approximation using $O(km(8/\epsilon)^\lambda)$ space, where $\lambda = O(D)$ for $D$-dimensional Euclidean metrics. Finally, for the special case of two groups, i.e., $m = 2$, we give an algorithm that obtains a $4(1+\epsilon)$-approximation using $O(k \epsilon^{-1} \log n)$ space for general metrics.

\smallskip
\noindent \textbf{Overview of \tauGMM}. We initialize the set of diverse points denoted by $\mathcal S$ with some arbitrary point in $\cX$. In each iteration, the algorithm selects an arbitrary point $p$ among the available points, denoted by $\mathcal R$, and includes it in $\mathcal S$, if the distance between the selected point $p$ and all other points in $\mathcal S$ is greater than $\tau = \gamma/2$. The point $p$ is made unavailable by removing it from $\mathcal R$ for subsequent iterations, irrespective of its inclusion in $\mathcal S$. Finally, the set $\mathcal S$ is returned.

\begin{algorithm}[t!]
		\caption{Algorithm \tauGMM }\label{algo:gammaGMM}
		{\small
			\begin{algorithmic}[1] 
				\Statex
				\begin{description}
					\item[\rlap{Input:}\phantom{Output:}] $\cX$: Universe of available points. 
					\item[\phantom{Output:}] $\tau \in \mathbb{R}^{+}$: a threshold on distance. 
					\item[\phantom{Output:}] $k \in \mathbb{Z}^{+}$.
					\item[\phantom{Output:}] $\cI:$ initialization set of points.
					\item[Output:] $\cS \subseteq \cX$ of size $k$.
				\end{description}	
				\State Let $\mathcal R \leftarrow \cX$ denote set of remaining points.
				\State If $\cI \neq \emptyset$, initialize $\mathcal{S} \leftarrow \cI$, otherwise $\cS \leftarrow \ $arbitrarily chosen point in $\cX$.
				\While{$|\mathcal{S}|< k$ \textbf{and} $\mathcal R \neq \emptyset$} 
					\State Let $p \in \mathcal{R}$ be an arbitrary point.
					\If{$\underset{v \in \mathcal{S}}{\text{min}} \ d(p, v)\geq \tau$}
					\State $\mathcal{S} \gets \mathcal{S} \cup \{p\}$.
					\EndIf
                \State $\mathcal R \leftarrow \mathcal{R} \setminus \{ p \}$.
				\EndWhile 
				\State \Return $\mathcal{S}$ 
			\end{algorithmic}
		}
\end{algorithm}

As defined in~\cite{ravi1994}, \fMM when $m = 1$ is also called \textsc{Max-Min} diversification. Using the theorem below, when our guess $\gamma$ is $\lsf$, we show that \tauGMM obtains a $2$-approximation for \textsc{Max-Min} diversification problem:
\begin{theorem}\label{thm:gammaGMM}
If $\gamma = \lsf$, the algorithm \tauGMM with $\tau = \gamma/2$ returns a $2$-approximation for the \textsc{Max-Min} diversification problem.
\end{theorem}
\begin{proof}
From the construction of set $\cS$, $d(p, q) \geq \tau = \gamma/2 = \lsf/2$ for every pair $p, q \in \cS$, we have the approximation guarantee. Now, we will argue that $|\cS| = k$. Let $S^* = \{ y_1, y_2, \cdots y_k \}$ denote the set of points that obtains optimal diversity score, i.e., $\diver(\cS^*)= \lsf$. We claim that for every point $p \in \cS$, there exists at most one point $f(p) \in \cS^*$ such that $d(p, f(p)) < \lsf/2$. For the sake of contradiction suppose there are two points $y, y' \in \cS^* \cap \rball{p}{\lsf/2}$ such that $d(p, y) < \lsf/2$ and $d(p, y') < \lsf/2$. From triangle inequality, we have $d(y, y') < \lsf$, which is not true because $d(y, y') \geq \diver(\cS^*) = \lsf$. For every point $p \in \cS$, consider $\rball{p}{\lsf/2}$. As $d(p, q) \geq \gamma/2 = \lsf/2$ for every pair $p, q \in \cS$, and $|\rball{p}{\lsf/2} \cap \cS^*| \leq 1$, there are $|\cS^*|$ balls containing exactly one point from $\cS^*$ and centered at a point from $\cS$. Therefore, $|\cS| \geq k$ and \tauGMM terminates when $|\cS| = k$. Hence, the theorem.
\end{proof}

\noindent Now, we describe an algorithm \tauGMMstream that is based on \tauGMM and is a data streaming algorithm. 

\smallskip
\noindent \textbf{Overview of \tauGMMstream}. The algorithm runs \tauGMM for each group $i \in [m]$ separately using the same threshold $\tau$. Therefore, when a point $p \in \cX_j$ in the stream is processed, \tauGMMstream checks if it can include it in the solution $\cS_j$ corresponding to group $j \in [m]$ using the threshold $\tau$. The algorithm returns $\bigcup_{j \in [m]} \cS_j$ containing $\tilk_j$ points of group $j$, for all $j \in [m]$. 

\begin{algorithm}[t!]
		\caption{Algorithm \tauGMMstream }\label{algo:tauGMMstream}
		{\small
			\begin{algorithmic}[1] 
				\Statex
				\begin{description}
					\item[\rlap{Input:}\phantom{Output:}] $\cX$: Universe of available points. 
					\item[\phantom{Output:}] $\tau \in \mathbb{R}^{+}$: a threshold on distance. 
					\item[\phantom{Output:}]  $\tilk_1, \ldots, \tilk_m \in \mathbb{Z}^{+}$.
					\item[Output:] $\cS \subseteq \cX$ of size $\tilk_1 + \tilk_2 + \cdots \tilk_m$. 
				\end{description}	
				\State Initialize $\cS_i \leftarrow \emptyset \ \forall i \in [m]$. 
				\While{processing point $p \in \cX_j \ \text{ for some } j \in [m]$ from stream} 
					\If{$|\cS_j| < \tilk_j$ and $\underset{v \in \mathcal{S}_j}{\text{min}} \ d(p, v)\geq \tau$}
					\State $\mathcal{S}_j \gets \mathcal{S}_j \cup \{p\}$.
					\EndIf
				\EndWhile 
				\State \Return $\mathcal{S} \leftarrow \bigcup_{i=1}^m 
				\mathcal{S}_i$. 
			\end{algorithmic}
		}
\end{algorithm}

\subsection{Data streaming algorithm for General metrics}\label{app:streamgen}
In this section, we describe a data stream algorithm called \textsc{Fair-Stream-Gen} that obtains a $30$-approximation for the \fMM problem. 

\smallskip
\noindent In~\cite{moumoulidou21}, the authors give a simple coreset construction for general metrics, by taking a union of the outputs of \GMM algorithm run on each group separately. They argue that by selecting \changes{at most} $k$ points using \GMM from each group, the resulting set of $O(k m)$ points in total, contains a $5$-approximation for the \fMM problem. Building on this, we give a data stream algorithm that first constructs the coreset using \tauGMMstream, and then, we use \randround~from Section~\ref{subsec:randomLP} to return a $6$-approximation on the coreset.

\smallskip
\noindent Let $d_\text{min} = \min_{p, q \in \cX} d(p, q) \ \text{ and } d_\text{max} = \max_{p, q \in \cX} d(p, q)$. We assume we are given a lower bound $\dmin^{\text{lb}}$ for $\dmin$, and an upper bound  $\dmax^{\text{ub}}$ for $\dmax$, following similar assumptions for fair $k$-center clustering in the data streaming setting~\cite{chiplunkar2020solve}. We use a parameter $\gamma$ to guess optimal diversity score $\lsf$. 

\smallskip
\noindent \textbf{Overview of \textsc{Fair-Stream-Gen}}. We identify the coreset using \tauGMMstream for a guess $\gamma$, we set $\tau = 2\gamma/5$ and $\tilk_j = k \ \forall j \in [m]$. The coreset $\cT_\gamma$ consists of \changes{$O(\tilk_j)$} points of group $j$, all separated by a distance of $2\gamma/5$. We use geometric guessing in the range $\Gamma :=  \{\dmin^{\text{lb}} , (1+ \epsilon)\dmin^{\text{lb}}, (1+\epsilon)^2\dmin^{\text{lb}}, \dots, \dmax^{\text{ub}}\}$ using the parameter $\gamma$ and run all these $O(\frac{1}{\epsilon}\log(\dmax^{\text{ub}}/\dmin^{\text{lb}}))$ guesses in parallel. 

Finally, we construct \fMM LP on the set of coreset points for all the guesses, given by, $\bigcup_{\gamma \in \Gamma} \cT_\gamma$. The fractional solution $x^*$ of this LP is then used to construct an integral solution using \randround~from Section~\ref{subsec:randomLP}.

\begin{algorithm}[t!]
		\caption{Algorithm \textsc{Fair-Stream-Gen} }\label{algo:streamgen}
		{\small
			\begin{algorithmic}[1] 
				\Statex
				\begin{description}
					\item[\rlap{Input:}\phantom{Output:}] $\cX$: Universe of available points.
					\item[\phantom{Output:}]  $k_1, \ldots, k_m \in \mathbb{Z}^{+}$.
					\item[\phantom{Output:}]  $\dmin^{\text{lb}}, \dmax^{\text{ub}}$: lower bound for minimum and upper bound \Statex \phantom{Output:} for maximum pairwise distances in $\cX$.
					\item[Output:] $\cS \subseteq \cX$.
				\end{description}	
			\For{each of the guesses $\gamma \in \Gamma :=  \{\dmin^{\text{lb}} , (1+ \epsilon)\dmin^{\text{lb}}, (1+\epsilon)^2\dmin^{\text{lb}}, \dots, \dmax^{\text{ub}}\}$, run in parallel}
			\State Set $\tilk_i = k \ \forall i \in [m].$
			\State $T_i \leftarrow \tau\textsc{-GMM-Stream}(\cX_i, 2\gamma/5, \tilk_1, \tilk_2, \cdots, \tilk_m) \ \forall i \in [m]$.
			\State $\cT_\gamma \leftarrow \bigcup_{i=1}^m T_i$.
			\EndFor
			\State Let $x^*$ be the optimal LP solution to \textsc{Fair Max-Min LP} for the points $\bigcup_{\gamma \in \Gamma} \cT_\gamma$. 
\State \Return $\cS \leftarrow \randround(\bigcup_{\gamma \in \Gamma} \cT_\gamma, x^*)$. 
			\end{algorithmic}
		}
\end{algorithm}

\begin{algorithm}[t!]
		\caption{Algorithm \textsc{Fair-Stream-Euclidean} }\label{algo:streameuc}
		{\small
			\begin{algorithmic}[1] 
				\Statex
				\begin{description}
					\item[\rlap{Input:}\phantom{Output:}] $\cX$: Universe of available points.
					\item[\phantom{Output:}]  $k_1, \ldots, k_m \in \mathbb{Z}^{+}$.
					\item[\phantom{Output:}] $\dmin^{\text{lb}}, \dmax^{\text{ub}}$: lower bound for minimum and upper bound \Statex \phantom{Output:} for maximum pairwise distances in $\cX$.
					\item[Output:] $\cS \subseteq \cX$.
				\end{description}	
			\For{each of the guesses $\gamma \in \Gamma := \{\dmin^{\text{lb}} , (1+ \epsilon)\dmin^{\text{lb}}, (1+\epsilon)^2\dmin^{\text{lb}}, \dots, \dmax^{\text{ub}}\}$, run in parallel}
			\State Set $\tilk_i = k(8/\epsilon)^{\lambda} \ \forall i \in [m] $ where $\lambda$ is doubling dimension of $\mathbb{R}^D$.
			\State $T_i \leftarrow \tau\textsc{-GMM-Stream}(\cX_i, \epsilon\gamma/4, \tilk_1, \tilk_2, \cdots, \tilk_m) \ \forall i \in [m]$.
			\State $\cT_\gamma \leftarrow \bigcup_{i=1}^m T_i$.
			\EndFor
			\State Let $x^*$ be the optimal LP solution to \textsc{Fair Max-Min LP} for the points $\bigcup_{\gamma \in \Gamma} \cT_\gamma$. 
\State \Return $\cS \leftarrow \randround(\bigcup_{\gamma \in \Gamma} \cT_\gamma, x^*)$. 
			\end{algorithmic}
		}
\end{algorithm}

Combining the $5$-approximation of coreset from~\cite{moumoulidou21} and the guarantees of \randround, we obtain the following theorem:
\begin{theorem}\label{thm:streamgen}
 The algorithm \textsc{Fair-Stream-Gen} returns a $\cS$ that obtains a $30(1+\epsilon)$-approximation for the diversity score and $1-\epsilon$-fairness for the $\fMM$ problem using $O\left(km \frac{1}{\epsilon}\log(\dmax^{\text{ub}}/\dmin^{\text{lb}}) \right)$ space. Here, $\dmax^{\text{ub}} \geq \max_{p, q \in \cX} d(p, q)$ and $\dmin^{\text{lb}} \leq \min_{p, q \in \cX} d(p, q)$  are given upper and lower bound estimates.
\end{theorem}
\begin{proof}
If $\gamma = \lsf/(1+\epsilon)$, following the same proof of Theorem 4 in~\cite{moumoulidou21}, we can argue that $\cT_\gamma$ is a $5(1+\epsilon)$-approximate coreset. Therefore, $\diver(\cT_\gamma) \geq \lsf/5(1+\epsilon)$. Combining this with the guarantees of Theorem~\ref{thm:6approx}, we get a $30(1+\epsilon)$-approximation and $1-\epsilon$ fairness.
\end{proof}

\subsection{Data streaming algorithm for Euclidean metrics}\label{app:streameuc}

In this section, we describe a data stream algorithm called \textsc{Fair-Stream-Euclidean} that obtains a $6$-approximation for \fMM. In Section~\ref{subsec:euclCoresets}, we gave an algorithm to construct a $(1+\epsilon)$-approximate coreset for Euclidean metrics, by taking a union of the outputs of \GMM algorithm run on each group. Using~\tauGMMstream, we obtain a $(1+\epsilon)$-approximate coreset. Similar to Section~\ref{app:streamgen}, we use \randround~from Section~\ref{subsec:randomLP} to return a $6$-approximation on the coreset.

\smallskip
\noindent Let $d_\text{min} = \min_{p, q \in \cX} d(p, q) \ \text{ and } d_\text{max} = \max_{p, q \in \cX} d(p, q)$. We assume we are given a lower bound $\dmin^{\text{lb}}$ for $\dmin$, and an upper bound  $\dmax^{\text{ub}}$ for $\dmax$, following similar assumptions for fair $k$-center clustering in the data streaming setting~\cite{chiplunkar2020solve}. We use a parameter $\gamma$ to guess optimal diversity score $\lsf$. 

\smallskip
\noindent \textbf{Overview of \textsc{Fair-Stream-Euclidean}}. We identify the coreset using \tauGMMstream for a guess $\gamma$, we set \changes{$\tau = \epsilon\gamma/4$} and \changes{$\tilk_j = k(8/\epsilon)^\lambda \ \forall j \in [m]$}, where $\lambda$ is doubling dimension of $\mathbb{R}^D$. The coreset $\cT_\gamma$ consists of $\tilk_j$ points of group $j$, all separated by a distance of $\epsilon\gamma/4$. We use geometric guessing in the range $\Gamma :=  \{\dmin^{\text{lb}}, (1+ \epsilon)\dmin^{\text{lb}}, (1+\epsilon)^2\dmin^{\text{lb}}, \dots, \dmax^{\text{ub}}\}$ using the parameter $\gamma$ and run all these $O(\frac{1}{\epsilon}\log(\dmax^{\text{ub}}/\dmin^{\text{lb}}))$ guesses in parallel. 

On the coreset obtained, we use \fDP, and the pairwise distances of points in the coreset are the new guesses for $\lsf$. We obtain the following guarantees:

\begin{theorem}\label{thm:streameuc}
The algorithm \textsc{Fair-Stream-Euclidean} returns a $\cS$ that obtains a $(1+\epsilon)$-approximation for the diversity score of $\fMM$ problem using $O\left(km(8/\epsilon)^{1+\lambda} \log(\dmax^{\text{ub}}/\dmin^{\text{lb}}) \right)$ space where $\lambda$ is doubling dimension of $\mathbb{R}^D$ and $\dmax^{\text{ub}} \geq \max_{p, q \in \cX} d(p, q)$, $\dmin^{\text{lb}} \leq \min_{p, q \in \cX} d(p, q)$ are given input upper and lower bound estimates. \end{theorem}
\begin{proof}
If $\gamma = \lsf/(1+\epsilon)$, following the same proof as Theorem~\ref{thm:coresets_doubling}, we can argue that $\cT_\gamma$ is a $(1+\epsilon)$-approximate coreset. Therefore, $\diver(\cT_\gamma) \geq \lsf/(1+\epsilon)$. Combining this with the guarantees of Theorem~\ref{thm:fairDP}, we get a $(1+\epsilon)$ approximation.
\end{proof}


\subsection{Data streaming algorithm when $m=2$} 
In this section, we describe a data streaming algorithm \streamrbcol which obtains a $4$-approximation to the \fMM problem when the number of groups $m = 2$. In~\cite{moumoulidou21}, the authors describe an algorithm called \textsc{Fair-Swap} which returns a $4$-approximation to the \fMM problem when the number of groups is $m = 2$. The algorithm can be extended to a $2$-pass streaming algorithm using $O(k)$ space with the same $4$-approximation guarantee. Therefore, our algorithm essentially reduces the number of passes with the same guarantees for the space used.  First, we describe the algorithm \textsc{Fair-Swap} from~\cite{moumoulidou21} briefly and then present an overview of our Algorithm~\ref{algo:stream2groups}.

\smallskip
\noindent Let $d_\text{min} = \min_{p, q \in \cX} d(p, q) \ \text{ and } d_\text{max} = \max_{p, q \in \cX} d(p, q)$. We assume we are given a lower bound $\dmin^{\text{lb}}$ for $\dmin$, and an upper bound  $\dmax^{\text{ub}}$ for $\dmax$, following similar assumptions for fair $k$-center clustering in the data streaming setting~\cite{chiplunkar2020solve}. We use a parameter $\gamma$ to guess optimal diversity score $\lsf$. 

The \textsc{Fair-Swap} algorithm described in~\cite{moumoulidou21} uses \GMM to retrieve a set $\cS$ of $k$ points in a group-agnostic way by ignoring the fairness constraints. Thus, it is possible that one group is under-represented (i.e., $|\cS \cap \cX_i| < k_i$) and the other is over-represented (i.e., $|\cS \cap \cX_i| > k_i$). So, the algorithm first identifies the under-represented group $i \text{ for some }i \in \{1, 2\}$. Then, it uses \GMM only on $\cX_i$, and obtains $\widetilde \cS_i$, called the \emph{swap-set} containing the remaining $k_i - |\cS \cap \cX_i|$ points. Finally for every point in $\widetilde \cS_i$ it removes the nearest neighbor in the over-represented group.

\smallskip
\noindent \textbf{Overview of Algorithm \streamrbcol}. We first note that for a guess $\gamma = \lsf$ we can \emph{simulate} the behavior of \gonzalez using \tauGMM algorithm given in the Appendix. In order to extend \textsc{Fair-Swap} to a single pass data streaming algorithm, we need to know the under-represented group. However, we can only determine it towards the end of running the \tauGMM algorithm. We overcome this by simultaneously treating
either of the two groups as under-represented whenever a new point in the
stream is selected by the~\tauGMM algorithm. In order to do that, we maintain $3$ sets $S, S_1, S_2$ using \tauGMM for all of them. In $S$, we include points in a group-agnostic way (similar to \textsc{Fair-Swap}) ignoring the fairness constraints. In $S_1$, we include points only of group $1$, and in $S_2$, we include points only of group $2$. By setting $\tau = \gamma/2$ we maintain the sets $S$, $S_1$ and $S_2$ such that all points are at least $\gamma/2$ distance apart in every one of them. 

Without loss of generality, suppose $\cX_1$ is the under-represented group. So, $|S \cap \cX_1| < k_1$. Our algorithm proceeds by identifying $k_1-|S \cap \cX_1|$ new points from $S_1$ by running \tauGMM with the initialization set of $S \cap \cX_1$ and $\tau = \gamma/4$. This ensures that the final set of points from group $1$ that are returned are $\gamma/4$ apart. By discarding the nearest neighbors of group $2$ from newly added points, we argue that our algorithm obtains a $4$-approximation. 

\begin{algorithm}[t!]
		\caption{Algorithm \streamrbcol: Data Stream Algorithm for $m=2$.}\label{algo:stream2groups}
		{\small
			\begin{algorithmic}[1] 
				\Statex
				\begin{description} 
				\item[\rlap{Input:}\phantom{Output:}] $\cX_1, \cX_2$: Universe of available points.
				\item[\phantom{Output:}] $k_1, k_2 \in \mathbb{Z}^{+}.$
				\item[\phantom{Output:}]  $\dmin^{\text{lb}}, \dmax^{\text{ub}}$: lower bound for minimum and upper bound for maximum pairwise distances in $\cX$.
				\item[Output:] Set $\cS$ of $k_1, k_2$ points in $\cX_1, \cX_2.$
			\end{description}
				\For{each of the guesses $\gamma \in \Gamma :=  \{\dmin^{\text{lb}} , (1+ \epsilon)\dmin^{\text{lb}}, (1+\epsilon)^2\dmin^{\text{lb}}, \dots, \dmax^{\text{ub}}\}$, run in parallel}
				\State Initialize $S  \leftarrow \emptyset$, $S_1  \leftarrow \emptyset$ and $S_2  \leftarrow \emptyset$.
				\While{processing point $p \in \cX_j$ from stream}
				\If{$d(p, S) \geq \gamma/2$ and $|S| < k$}
				\State $S \leftarrow S \cup \{ p \}.$
				\EndIf
				\If{$j = 1$ and $d(p, S_1) \geq \gamma/2$ and $|S_1| < k_1$}
				\State $S_1 \leftarrow S_1 \cup \{ p \}.$
				\EndIf 
				\If{$j = 2$ and $d(p, S_2) \geq \gamma/2$ and $|S_2| < k_2$}
				\State $S_2 \leftarrow S_2 \cup \{ p \}.$
				\EndIf
				\EndWhile
				\State $T_j \leftarrow S \cap \cX_j$ for $j \in \{ 1, 2 \}$.
				\State Set $ u \leftarrow \text{arg min}_j |T_j| - k_i$ and $ o \leftarrow 3-u$.
				\State $E_u  \leftarrow \ $\tauGMM$(S_u, \gamma/4, k_u - |T_u|, T_u)$
				\State $R_o \leftarrow \{ \text{arg min}_{q \in T_o} d(p, q) : p \in E_u \setminus T_u \}$. 
				\State $\cS_\gamma \leftarrow  E_u \cup (T_o \setminus R_o)$.
				\EndFor
				\State \Return $\cS \leftarrow \underset{\cS_\gamma : \ \gamma \in \Gamma}{\text{arg max}} \ \diver(\cS_\gamma)$
			\end{algorithmic}
		}
\end{algorithm}

\smallskip
\noindent We obtain the following guarantees for \streamrbcol:

\begin{theorem}[Theorem~\ref{thm:stream2group_informal} restated]
There is a data streaming algorithm that obtains a  $4(1+\epsilon)$-approximation for \fMM problem using $O(k\epsilon^{-1}\log(\dmax^{\text{ub}}/\dmin^{\text{lb}}))$ space, where $\dmax^{\text{ub}} \geq \max_{p, q \in \cX} d(p, q)$ and $\dmin^{\text{lb}} \leq \min_{p, q \in \cX} d(p, q)$ are given input upper and lower bound distance estimates.
\end{theorem}
\begin{proof}
Consider the optimal solution $\cS^* = \{ y_1, y_2, \cdots y_k \}$ where $k = k_1 + k_2$. Consider set $S_1$ after the stream has ended. We can observe that for every point $p \in S_1$, $|\rball{p}{\lsf/2} \cap \cS^*| \leq 1$. Therefore, $|S_1| = k_1$. Similarly, we can argue that $|S| = k$ and $|S_2| = k_2$. 

Let $T_i = S \cap \cX_i$ for $i \in \{1, 2\}$ and $u \in \{1, 2\}$ denote the under-represented group, i.e., $|T_u| < k_u$ and $o$ denote the over-represented group, i.e., $|T_o| > k_o$. This means that we need to find  $k_u - |T_u|$ additional points to add to $S$ to satisfy the fairness constraint for group $u$. First, we argue that $S_u$ contains $k_u - |T_u|$ points that are at least
$\gamma/4$ distance from all the points in $T_u$. As $\tau = \gamma/2$, we have $\diver(S_u) \geq \gamma/2$ and $\diver(T_u) \geq \gamma/2$. Consider balls of radius $\gamma/4$ centered around points in $T_u$, given by $\rball{p}{\gamma/4} \ \forall p \in T_u$. Using triangle inequality, we have that $|\rball{p}{\gamma/4} \cap S_u| \leq 1 \ \forall p \in T_u$, as otherwise, we would have two points in $S_u$ with distance strictly less than $\gamma/2$. This implies that there are at least $k_u - |T_u|$ points in $S_u$,  denoted by $\widetilde S_u$ with distance at least $\gamma/4$ from points in $T_u$. As $\widetilde S_u \subseteq S_u$, we have $\diver(\widetilde S_u) \geq \gamma/2$. Therefore, using \tauGMM on the set $S_u$ with $\tau = \gamma/4$ and initialized with points $T_u$, we obtain the set $E_u$ containing $k_u$ points such that $\diver(E_u) \geq \gamma/4$.

We remove the nearest neighbors denoted by the set $R_o$ from $\cX_o$ of the newly included points $E_u \setminus T_u$. As $\diver(T_o) \geq \gamma/2$, we have for $\diver(R_o) \geq \gamma/2$. Using triangle inequality, we have that $|\rball{p}{\gamma/4} \cap R_o| \leq 1 \ \forall p \in E_u \setminus T_u$, as otherwise, we would have two points in $R_o$ with distance strictly less than $\gamma/2$. Therefore, the points remaining given by $\cS_\gamma = E_u \cup (T_o \setminus R_o)$ satisfy $\diver(\cS_\gamma) \geq \gamma/4$. For the guess $\gamma \geq \lsf/(1+\epsilon)$, we have that: $\diver(
\cS_\gamma) \geq \frac{\lsf}{4(1+\epsilon)}.$

As the total number of points stored is $O(|S| + |S_1| + |S_2|) = O(k)$, the theorem follows. 
\end{proof}

\section{Scalable Algorithms: Distributed Systems \& Composable Coresets}\label{app:composable_coresets}

In this section we show the approach discussed in Section~\ref{sec:composable_coresets} yields an $(1+\epsilon)$-composable coreset $\cT$ for \fMM problem. Our proof is similar to Theorem~\ref{thm:coresets_doubling} and the main difference is in identifying
\emph{critical} and \emph{non-critical} groups based on the properties of  $T_j \ \forall j \in [m]$.

The analysis uses the \emph{anti-cover} property, which every set selected by \gonzalez satisfies. Let $S = \{s_1, s_2, \dots, s_k\} \subseteq \cX$ be a set of size $k$ selected by \gonzalez. Then, we say $S$ is an \emph{$r$-net} for $\cX$, or equivalently $S$ satisfies the \emph{anti-cover} property, if for $r = \min\limits_{u \in \left\{S \setminus s_k \right\}} d(u, s_{k}) = \diver{(S)} \ ,$ the following properties hold: (1)~(\emph{separation}) for any $u, v \in S, d(u, v) \geq r$, and (2)~(\emph{coverage}) for any $v \in \cX \setminus S, \min\limits_{u \in S} d(u, v) \leq r$. We now prove the following result:

\begin{theorem}[Theorem~\ref{thm:comp_coresets_doubling} restated] $\cT$ is a $(1+\epsilon)$-composable coreset for fair Max-Min diversification of size $O((8/\epsilon)^\doubldim km L)$ in metrics of doubling dimension $\doubldim$ which can be obtained in $O((8/\epsilon)^\doubldim kmnL)$ time.
\end{theorem}

\begin{proof}
We show how by analyzing the properties of the $T_j$ sets, we can categorize
the groups into \emph{critical} and \emph{non-critical}. Then the ideas in
Theorem~\ref{thm:coresets_doubling} in the sequential setting apply here as
well.

Let $T_i$ be the set of points of group $i$ present in $\cT$, namely for a
fixed value $i$: $T_i \leftarrow \bigcup_{j=1}^{L} T_j^i$. For any point $x \in
\cX_i$ for $i \in [m]$, define $f(x) = \min\limits_{y \in T_i} d(x, y)$. By the \emph{anti-cover} property of \gonzalez, each $T_j^i \subset T_i$ set of
points that belong to group $i$ is an $r_j^i-$ net for $\cY_j^i$, where
$r_j^i = \diver{(T_j^i)}$. Then $T_j^i$ satisfies the \emph{separation} and
\emph{coverage} properties, thus the following hold: 

\begin{enumerate} 
\item For any pair of points $x, y \in T_j^i$, it holds
that $d(x, y) \geq r_j^i$. 
\item For any point $x \in \cY_j^i\setminus T_j^i$,
it holds that $d(x, f(x)) \leq r_j^i$. 
\end{enumerate}

Now, following the analysis of Indyk at al.~\cite{Indyk14},
define $r^i = \max\limits_{j \in [L]} \ r_j^i$. Then by the coverage property
it follows that $d(x, f(x)) \leq r^i$ for any point $x \in \cX_i$ in group $i$
for $i\in [m]$. Also notice that if $|T_j^i| < (4/\foo)^\doubldim k$ for all $j
\in [L]$, then $r^i$ can be treated as zero since $d(x, f(x)) = 0$ for all
points in $\cX_i$. Further, since $T_i$ is a superset of the $T_j^i$ sets, it contains a set $T'_i \subseteq T_i$ with $\diver{(T'_i)} \geq r^i$. Next, notice that we can define the following cases for the value of $r^i$ to see if a group is critical or not:

\begin{enumerate}
\item If $r^i \geq (\foo/2) \divScore$ group $i$ is \emph{non-critical}. Also
notice there exists a set in $T_i$ with at least $(4/\foo)^\doubldim k$
points that are greater or equal than $(\foo/2) \divScore$ apart from each
other.

\item If $r^i < (\foo/2) \divScore$ group $i$ is \emph{critical}. Also notice
that for any point $x \in \cX_i$ it holds that $d(x, f(x)) \leq r^i < (\foo/2)
\divScore$.
\end{enumerate}

Finally, using similar arguments as in the sequential setting in
Theorem~\ref{thm:coresets_doubling}, we can prove that $\cT$ is a
$(1+\epsilon)-$\emph{composable} coreset.
\end{proof}

We now discuss a simple two-round distributed algorithm for \fMM that uses composable coresets.

\medskip

\noindent
\textbf{A two-round distributed algorithm using composable coresets.} We assume the data is partitioned into $L$ processing sites. A similar approach based on coresets was proposed by~\cite{Ceccarello17} for the unconstrained Max-Min diversification problem. In the first round, each site $j \in [L]$ computes its local coreset $T_j$ and sends it to the coordinator site. In the second round, at the coordinator site we use the best known approximation algorithm in the sequential setting or a brute force approach. Here, in the latter case we  retrieve an $(1+\epsilon)$-approximate solution. Alternatively, we could use the linear programming approach discussed in Section~\ref{sec:genmetrics}. In that case, we get a $6(1+\epsilon)$-approximate solution for the \fMM problem by only sacrificing at most $\epsilon k_i$ points per group $i$. Similar arguments follow if we are in general metric spaces. 

\end{document}